%% file: cross-engine.tex
\documentclass[sigconf,authorversion,nonacm]{acmart}
\usepackage{graphicx}
\usepackage{microtype}
\usepackage{epstopdf}
\usepackage{flushend}
\usepackage{balance}
\usepackage{xcolor}
\usepackage{verbatim}
\usepackage{listings}
\usepackage{booktabs}
\usepackage{nicefrac}
\usepackage{paralist}
\usepackage{caption}
\usepackage{subcaption}
\usepackage{soul}
\usepackage{multirow}
\usepackage{enumitem}
\usepackage{tabularx}
\usepackage{xspace}
\usepackage{verbatim}
\usepackage{pdfpages}
\newcommand{\skeena}{{Skeena}\xspace}

\newcommand{\ermia}{\texttt{ERMIA}\xspace}
\newcommand{\ermias}{\texttt{ERMIA-S}\xspace}
\newcommand{\innodb}{\texttt{InnoDB}\xspace}
\newcommand{\innodbs}{\texttt{InnoDB-S}\xspace}
\newcommand{\innodbms}{\texttt{InnoDB-MS}\xspace}
\newcommand{\innodbm}{\texttt{InnoDB-M}\xspace}
\newcommand{\tmpfs}{\texttt{tmpfs}\xspace}

\definecolor{comment-red}{rgb}{0,0,0}
\definecolor{edit-red}{rgb}{0,0,0}
\definecolor{OliveGreen}{rgb}{0,0.6,0}
\newcommand{\edit}[1]{\textnormal{#1}\unskip}

\def\thepaperkeywords{Cross-engine transactions; snapshot isolation; OLTP; MySQL}
\def\thepapertitle{\skeena: Efficient and Consistent Cross-Engine Transactions}


\hypersetup{
  pageanchor=true,
  plainpages=false,
  pdfpagelabels,
  bookmarks,
  bookmarksnumbered,
  pdfborder=0 0 0,  
  colorlinks=true,
  linkcolor=ACMRed,
  anchorcolor=ACMRed,
  citecolor=Blue, 
  pdftitle={\thepapertitle},
  pdfauthor={Jianqiu Zhang, Kaisong Huang, Tianzheng Wang and King Lv},
  pdfsubject={},
  pdfkeywords={\thepaperkeywords},
}

\usepackage{tikz}
\newcommand*\circled[1]{\tikz[baseline=(char.base)]{
            \node[shape=circle,fill=.,inner sep=0pt] (char) {\color{-.}\textsf\footnotesize #1};}}

\usepackage{breakurl}
\usepackage{algorithm}

\copyrightyear{2022} 
\acmYear{2022} 
\setcopyright{acmcopyright}
\acmConference[SIGMOD '22]{Proceedings of the 2022 International Conference on Management of Data}{June 12--17, 2022}{Philadelphia, PA, USA}
\acmBooktitle{Proceedings of the 2022 International Conference on Management of Data (SIGMOD '22), June 12--17, 2022, Philadelphia, PA, USA}
\acmPrice{15.00}
\acmISBN{978-1-4503-9249-5/22/06}
\acmDOI{10.1145/3514221.3526171}


\newcommand{\FAA}{\texttt{FAA}\xspace}

\definecolor{Blue}{rgb}{0, 0, 0}
\definecolor{comment-red}{rgb}{0,0,0}
\definecolor{OliveGreen}{rgb}{0,0.6,0}
\definecolor{BrickRed}{rgb}{0, 0, 0}

\newcommand{\xdep}{\rightarrow}
\newcommand{\txdepimpl}[3]{\xrightarrow{\mathit{#1{:}#2#3}}}
\newcommand{\txdep}[3][]{%
  \ifthenelse{\equal{#1}{}}
  {\txdepimpl{#2}{}{#3}}
  {\txdepimpl{#2}{#1{:}}{#3}}
}

\begin{document}
\fancyhead{}


\title{\thepapertitle}
\subtitle{To appear at SIGMOD 2022}

\lstset {
language=C,
basicstyle=\ttfamily\small,
keywordstyle=\ttfamily\bfseries\color{blue},
commentstyle=\color{OliveGreen},
numbers=left,
numberstyle=\small,
numbersep=5pt,
tabsize=1,
gobble=0,
stepnumber=2,
xleftmargin=15pt,
escapeinside={(@*}{*@)},
morekeywords={},
columns=fullflexible,
}

\author{Jianqiu Zhang}
\email{jianqiuz@sfu.ca}
\affiliation{%
  \institution{Simon Fraser University}
  \country{}
}

\author{Kaisong Huang}
\email{kha85@sfu.ca}
\affiliation{%
  \institution{Simon Fraser University}
  \country{}
}

\author{Tianzheng Wang}
\email{tzwang@sfu.ca}
\affiliation{%
  \institution{Simon Fraser University}
  \country{}
}

\author{King Lv}
\email{lvjinquan@huawei.com}
\affiliation{
  \institution{Huawei Cloud Database Innovation Lab}
  \country{}
}

%

\input{0-abstract}

\maketitle
\pagestyle{empty}

\renewcommand\thetable{\arabic{table}}    
\renewcommand\thefigure{\arabic{figure}}    
\setcounter{figure}{0}   
\setcounter{table}{0}   
\setcounter{section}{0}
\setcounter{page}{1}

\input{1-introduction}

\input{2-background}

\input{3-design}

\input{4-mysql}

\input{5-evaluation}

\input{6-related-work}

\input{7-conclusion}

\begin{acks}
We thank Chong Chen, Per-{\AA}ke Larson, Qiang Liu, Chengwei Zhang, Zongquan Zhang, Yanhui Zhong and Qingqing Zhou for their valuable discussions and comments on this project. 
We also thank the reviewers for their constructive feedback. 
This work received support from Huawei Cloud Database Innovation Lab.
\end{acks}

\balance
\bibliographystyle{ACM-Reference-Format}
\bibliography{ref}

\end{document}

%% file: 0-abstract.tex
\begin{abstract}
Database systems are becoming increasingly multi-engine. 
In particular, a main-memory database engine may coexist with a traditional storage-centric engine in a system to support various applications.
It is desirable to allow applications to access data in both engines using {\it cross-engine} transactions. 
But existing systems are either only designed for single-engine accesses, or impose many restrictions by limiting cross-engine transactions to certain isolation levels and table operations. 
The result is inadequate cross-engine support in terms of correctness, performance and programmability. 

This paper describes \skeena, a holistic approach to cross-engine transactions. 
We propose a lightweight snapshot tracking structure and an atomic commit protocol to efficiently ensure correctness and support various isolation levels. 
Evaluation results show that \skeena maintains high performance for single-engine transactions and enables cross-engine transactions which can improve throughput by up to 30$\times$ by judiciously placing tables in different engines. 
\end{abstract}

%% file: 1-introduction.tex
\section{Introduction}
\label{sec:intro}
Traditional database engines are storage-centric: they assume data is storage-resident and optimize for storage accesses. 
Modern database servers often feature large DRAM that fits the working set or entire databases, 
enabling memory-optimized database engines~\cite{MMDBOverview,Cicada,Silo,FOEDUS,ERMIA,Deuteronomy,Hekaton,HStore,Hyper,MOT} that perform drastically better with lightweight concurrency control, indexing and durability designs. 

Now suppose you are a database systems architect, and inspired by recent advances, built a new memory-optimized engine. 
But soon you found it was difficult to attract users: 
some do not need such fast speed; some say ``I want it only for some tables or part of my application.''
A common solution is to integrate the new engine into an existing system that initially uses a traditional engine. 
The result is a \textit{multi-engine} database system (Figure~\ref{fig:meds}). 
The application can judiciously use tables in both engines. 
Although engines share certain services (e.g., SQL parser), each engine is autonomous with its own indexes, concurrency control, etc. 
Some systems~\cite{HekatonWhitepaper,MySQLAltEngines,PGFDW}
already take this approach for easier migration and compatibility. 

\begin{figure}[t]
	\centering
	\includegraphics[width=\columnwidth]{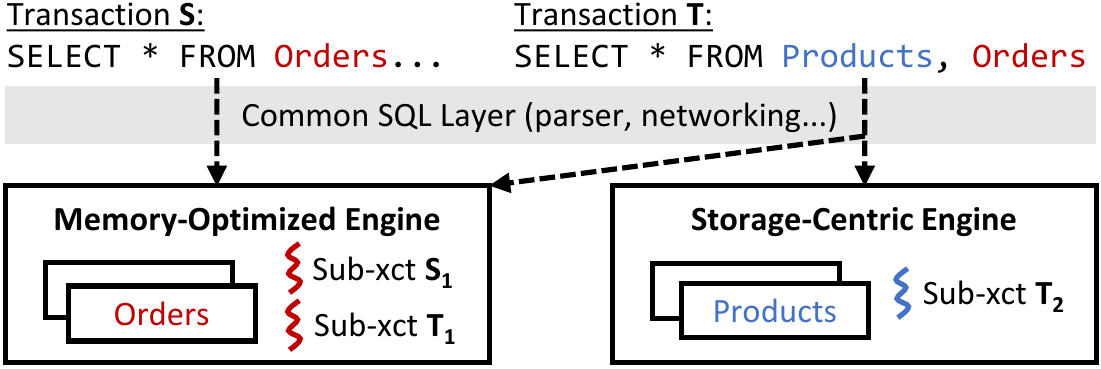}
	\caption{\textmd{Multi-engine database system. 
  Data accesses are routed to the corresponding storage engines.}}
	\label{fig:meds}
\end{figure}

\subsection{Cross-Engine: a Poorly-Supported Necessity}
\label{subsec:intro-cross}
As an experienced architect---perhaps even before users did---you realized it was necessary to support cross-engine transactions. 
For example, a financial application may use a memory table for fast trading and keep other data in the traditional engine for low cost; yet the user may need to access both engines for recent and historical trading data in one ACID transaction~\cite{SQLServerCET}. 
The application may use a unified SQL interface to access all engines, 
but since each engine is implemented as an autonomous ``package,'' the system has to use each engine's own transaction abstractions; we refer to them as \textit{sub-transactions}. 
A \textit{transaction} consists of at least one sub-transaction. 
In Figure~\ref{fig:meds}, $S$ is single-engine with $S_1$, while $T$ is cross-engine with $T_1$ (memory-optimized) and $T_2$ (storage-centric). 

Cross-engine transactions can be very useful, but existing support is inadequate in terms of correctness, performance and programmability. 
First, although simply starting and committing sub-transactions suffice to support single-engine transactions, doing so does not ensure correct cross-engine execution. 
A transaction over two engines that both use snapshot isolation (SI)~\cite{ANSICritique} can still see inconsistent data and run under a lower-than-SI isolation level. 
Even if both engines ensure serializability, the overall execution is not necessarily serializable. 
Simply committing sub-transactions also risks atomicity if a sub-transaction fails to commit. 

Second, prior designs are not aimed at achieving high performance in modern multi-engine systems, which are \textit{fast-slow} where a (much faster) memory-optimized engine and a (much slower) storage-centric engine coexist in a single node. 
So it is vital for the cross-engine solution to impose low (if any) overhead, especially on the faster engine. 
Prior solutions~\cite{MYRIAD,FederatedSI,GeneralizedSI,FDBSSurvey,MDBSOverview,IncrementalDSI,CO} ignored this hidden requirement by assuming a cluster of similar systems. 

Finally, past solutions are often at odds with (1) keeping engine autonomy for maintainability as engines are typically developed by different teams (but still of the same vendor), and (2) easing application development. 
They often require non-trivial application changes and limit functionality, by forcing users to pre-declare whether a transaction is cross-engine or to use certain isolation levels~\cite{HekatonWhitepaper}; both can be complex and affect performance. 

\subsection{\skeena}
We present \skeena, a holistic approach to efficient and consistent cross-engine transactions in the context of multi-versioned, fast-slow systems. 
We make three key observations to guide \skeena's design. 
First, as noted by prior work~\cite{IncrementalDSI}, inconsistent snapshots can be avoided by carefully selecting a snapshot in each engine. 
This requires efficiently tracking snapshots that can be safely used by later transactions. 
Second, in addition to using correct snapshots and enforcing sub-transactions commit in the same order across engines, for serializability it suffices to require each engine use commit ordering, i.e., forbid schedules where commit and dependency orders mismatch~\cite{CO,WeakConsistency}. 
Many concurrency control protocols exhibit this property, including the widely-used 2PL and optimistic concurrency control (OCC)~\cite{Silo,Hekaton,OCC,Cicada,HyperSI}. 
Finally, engines are developed and/or well understood by the same vendor, potentially allowing non-intrusive changes to engines for more optimizations. 

Based on these observations, we design \skeena to consist of (1) a cross-engine snapshot registry (CSR) for correct and efficient snapshot selection and (2) an extended pipelined commit protocol for atomicity and durability without (expensive) traditional 2PC. 
\skeena can be easily plugged into an existing system. 

Conceptually, CSR maintains mappings between commit timestamps (therefore snapshots) in one engine and those in another. 
A transaction may start by accessing any engine using the latest snapshot $s$. 
Upon accessing another engine $E$, it queries CSR using $s$ to select a snapshot in $E$ using which would avoid incorrect executions. 
Further, with CSR one only needs to set each engine to use a serializable protocol that exhibits commit ordering to guarantee serializability. 
Later, we discuss the detailed algorithms to realize this idea and techniques that make CSR lightweight and easy to maintain. 
In fast-slow systems, CSR incurs negligible overhead as the storage accesses in the traditional engine present a bigger bottleneck, and single-engine transactions do not access CSR at all. 

Leveraging the fact that engines can communicate via fast shared memory (e.g., in the same address space), \skeena extends the widely-used group/pipelined commit protocols~\cite{Aether,NVM-DLog,TaurusLog} to ensure atomicity and durability.  
Upon commit, the worker thread detaches the transaction and places it on a commit queue, before continuing to work on the next request. 
Meanwhile, a background committer thread monitors the queue and durable log sequence numbers in both engines to dequeue transactions whose sub-transactions' log records have been fully persisted. 
This way, \skeena ensures cross-engine transactions are not committed (i.e., with results made visible to the application) until all of its sub-transactions are committed, while avoiding expensive 2PC.

We adopted \skeena in MySQL to enable cross-engine transactions across its default InnoDB and ERMIA~\cite{ERMIA}, an open-source main-memory OLTP engine. 
This required 83 LoC out of over 200k LoC of the entire codebase. 
Evaluation 
on a 40-core server shows that \skeena retains the memory-optimized engine's high performance, and incurs very low additional overhead for cross-engine transactions. 
By judiciously placing tables in both engines, \skeena can help improve the throughput of realistic workloads by up to 30$\times$ compared to using traditional engines. 

Note that our goal is \textit{not} to build faster database engines, nor to invent new concurrency control protocols for cross-engine transactions; both are well studied by prior work.
Instead, we aim to (1) enable cross-engine transactions without excessive overhead and (2) explore practical designs for modern fast-slow systems. 

\subsection{Contributions} 
This paper makes five contributions. 
\circled{1} We analyze the correctness requirements of cross-engine transactions under various isolation levels, ranging from read committed to serializable. 
\circled{2} We distill a set of desirable properties and design principles to be followed by multi-engine systems.
\circled{3} We propose \skeena, a holistic approach to consistent cross-engine transactions by leveraging the properties of the fast-slow multi-engine architecture.
\circled{4} We show \skeena's feasibility and explore practical design issues by integrating an open-source memory-optimized engine (ERMIA) into MySQL alongside its storage-centric engine (InnoDB).
\circled{5} Through comprehensive experiments, we explore the potential and distill useful recommendations of using cross-engine transactions to improve performance and reduce storage costs under realistic workloads. 
\skeena is open-sourced at \url{https://github.com/sfu-dis/skeena}.

%% file: 2-background.tex
\section{Background} 
\label{sec:bg}
In this section, we give the necessary background for cross-engine transactions and motivate our work.

\subsection{Modern Fast-Slow Multi-Engine Systems}
\label{subsec:fast-slow}
We have described the idea of multi-engine systems in Section~\ref{sec:intro}.
Several production systems already adopted the fast-slow architecture: 
SQL Server supports memory-optimized tables managed by its Hekaton main-memory engine~\cite{Hekaton,SQLServerMemTable}; 
PostgreSQL supports additional engines through foreign data wrapper~\cite{PGFDW}, which is used by Huawei GaussDB to integrate a main-memory engine~\cite{MOT}. 

Multi-engine systems bear similarities to distributed and federated database systems~\cite{FDBSSurvey,MDBSOverview,SuperDB,MYRIAD,MYRIADDesign,IncrementalDSI,GeneralizedSI,1SRSI,LazySI}, but are unique in several ways. 
As Table~\ref{tbl:meds} summarizes, a multi-engine system integrates engines developed and/or understood by the same vendor;
in contrast, federated systems consist of opaque systems developed by different vendors. 
Distributed systems typically involve a set of nodes that run the same engine carefully designed to support distributed transactions, exhibiting low autonomy. 
Fast-slow systems integrate different engines that vary in performance, so an inefficient cross-engine solution may penalize single-engine transactions, defeating the purpose of adopting a fast engine; 
mitigating such overhead is the major goal of our work. 
Note that multi-engine systems often allow slightly trading autonomy for performance and compatibility, e.g., by managing schemas in all engines centrally. 
However, federated systems allow little room for doing so, as each system is usually a proprietary package. 
Multi-engine systems can scale up and out, whereas the other two types of systems mainly focus on scaling out.
We focus on single-node fast-slow systems and leave scaling out as future work.
Finally, both multi-engine and federated systems may present applications with a unified interface (e.g., SQL). 
But this does not automatically guarantee correctness for cross-engine transactions. 
Unlike federated systems which already address this issue~\cite{FDBSSurvey,MDBSOverview,MYRIAD}, modern fast-slow systems either completely lack the support for cross-engine transactions (e.g., MySQL) or come with many restrictions (e.g., SQL Server); 
we elaborate in Section~\ref{subsec:motivation} after introducing more necessary background next. 

\subsection{Database Model and Assumptions}
\label{subsec:model}

Now we lay out the preliminaries for analyzing cross-engine transactions in fast-slow systems. 

\begin{table}[t]
\caption{\textmd{Multi-engine vs. distributed and federated systems.}}
\label{tbl:meds}
\begin{tabular}{@{}p{2.07cm}p{1.91cm}p{1.75cm}p{1.7cm}@{}}
\toprule
\textbf{} & \textbf{Multi-Engine}  & \textbf{Federated} & \textbf{Distributed} \\ \midrule
{Engine Internals} & Transparent & \textit{Opaque} & Transparent \\ \midrule
{Engine Types} & Heterogeneous & Heterogeneous & \textit{Homogeneous} \\ \midrule
{Autonomy} & \textit{Almost full} & Full & Low \\ \midrule
{Scalability} & \textit{Up and/or out} & Out & Out \\ \bottomrule
\end{tabular}
\end{table}

\textbf{Multi-Versioning.} 
Many popular systems are multi-versioned, including storage-centric (e.g., MySQL InnoDB, PostgreSQL and SQL Server) and memory-optimized (e.g., Hekaton~\cite{Hekaton}, ERMIA~\cite{ERMIA} and Cicada~\cite{Cicada}) engines.
Given the wide adoption, we focus on multi-versioned systems.
Following prior work~\cite{SSI,IncrementalDSI,SSN,WeakConsistency,GeneralizedIsolationLevels}, 
we model databases as collections of records, each of which is a totally-ordered sequence of \textit{versions}. 
Updating a record appends a new version to the record's sequence. 
Inserts and deletes are special cases of updates that append a valid and special ``invalid'' version, respectively.
Obsolete versions (as a result of deletes/updates) are physically removed only after no transaction will need them, using reference counting or epoch-based memory management~\cite{ERMIA,Steam}.

Reading a record requires locating a proper version; this is dictated by the concurrency control protocol. 
We base on a common design~\cite{Hekaton,ERMIA,Deuteronomy,MVCCEval} where the engine maintains a global, monotonically increasing counter that can be atomically read and incremented. 
Note that in a multi-engine system, engines maintain their own timestamp counters; 
for now, we assume single-engine transactions and expand to cross-engine cases later. 
Each transaction is associated with a begin timestamp and a commit timestamp, both drawn from the counter. 
Upon commit, the transaction obtains its commit timestamp, which determines its commit order by atomically incrementing the counter. 
Each version is associated with the commit timestamp of the transaction that created it. 
Transactions access data using a \textit{snapshot} (aka \textit{read view}), which is a timestamp that represents the database's state at some point in (logical) time. 

\textbf{Isolation Levels.}
For read committed, we always read the latest committed version. 
SI allows the transaction to read the latest version created before its begin timestamp obtained upon transaction start or the first data access. 
A transaction can update a record if it can see the latest committed version. 
Serializability can be achieved by locking~\cite{DistributedROTx,IntegratedCC,DistTxImpl} and certifiers~\cite{SSI,SSN,OCC}.
We aim to enforce these isolation levels in the presence of cross-engine transactions. 

\textbf{Cross-Engine ACID Properties.}
Analogous to maintaining ACID properties in a single engine, a multi-engine system must maintain these for both single- and cross-engine transactions:

\begin{itemize}[leftmargin=*]\setlength\itemsep{0em}
\item Atomicity: All the sub-transactions should eventually reach the same commit or abort conclusion, i.e., either all or none of the sub-transactions commit in their corresponding engines. 
\item Consistency: All transactions (single- or cross-engine) should transform the database from one consistent state to another, enforcing constraints within and across engines. 
\item Isolation: Changes in any engine made by a cross-engine transaction must not be visible until the cross-engine transaction commits, i.e., all sub-transactions have committed.
\item Durability: Changes made by cross-engine transactions should be persisted while guaranteeing atomicity.
\end{itemize}
Enforcing cross-engine ACID requires careful coordination of sub-transactions to avoid anomalies, as we describe next.

\subsection{Cross-Engine Anomalies}
\label{subsec:anomalies}
The relative ordering of sub-transaction begin/commit events directly determines correctness, as certain ordering may lead to anomalies and violate ACID requirements, as described next. 

\begin{figure}[t]
\centering
\includegraphics[width=0.95\columnwidth]{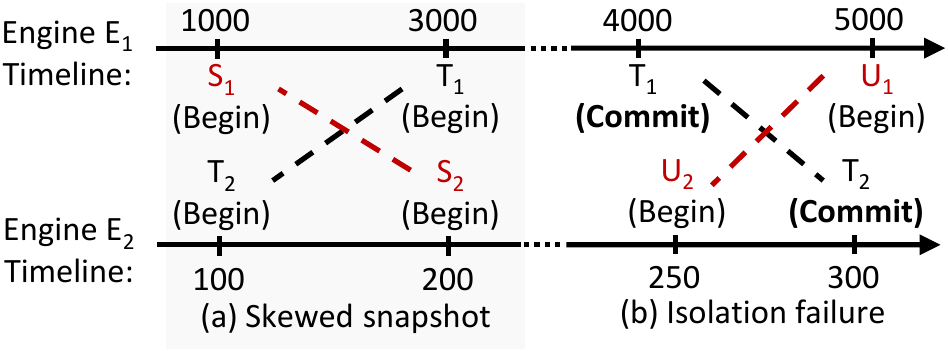}
\caption{\textmd{Inconsistent snapshots. (a) $S$ uses an older (newer) snapshot in $E_1$ ($E_2$). 
(b) $U$ sees $T_1$'s results, but does not see $T_2$'s.}}
\label{fig:anomalies}
\end{figure}

\textbf{Issue 1: Inconsistent Snapshots.}  
There are two cases where a transaction may be given an inconsistent snapshot. 
In Figure~\ref{fig:anomalies}(a), $S$ started in $E_1$ with snapshot 1000, and $T$ started in $E_2$ with snapshot 100.
Suppose another transaction in $E_1$ committed by incrementing $E_1$'s timestamp counter to 3000. 
Then, $T$ accesses $E_1$, which assigns $T_1$ its latest snapshot 3000, and $S_2$ obtains snapshot 200 in $E_2$.
Compared to $S$, $T$ sees a newer version of the database in $E_1$, but an older version in $E_2$.
This would require $S$ and $T$ to start before each other, which is impossible in a correct SI schedule~\cite{WeakConsistency}. 
This corresponds to the ``cross'' phenomenon in distributed SI (DSI)~\cite{IncrementalDSI}. 

Another anomaly (isolation failure) may allow transactions to see partial results. 
In Figure~\ref{fig:anomalies}(b), $T$ first commits $T_1$ with timestamp 4000. 
Until $T_2$ is fully committed, $T$ is still in-progress, so none of its changes should be visible to other transactions. 
Meanwhile, $U$ starts in $E_2$ with timestamp 250, and continues to open $U_1$: since $U_1$ started after $T_1$ committed, by definition it should see $T1$'s changes. 
Thus, $U$ sees partial results: $T$'s results are visible in $U$'s snapshot in $E_1$, but not $E_2$.
This anomaly corresponds to the serial-concurrent phenomenon in DSI~\cite{IncrementalDSI}. 
Compared to inconsistent (skewed) snapshots which concern the order in which sub-transactions are opened, isolation failure arises when sub-transaction begin and commit actions are interleaved and inflict different write-read dependency orders in different engines. 

\textbf{Issue 2: Serializability.}
Even if both engines guarantee full serializability, the overall execution may not be serializable. 
As Figure~\ref{fig:serial-dep}(a) shows, $S$ and $T$ are concurrently executing in two engines that offer serializability. 
Each engine runs a serializable schedule, with an anti-dependency shown in Figure~\ref{fig:serial-dep}(b). 
However, as shown in Figure~\ref{fig:serial-dep}(c), the overall execution exhibits write skew with cyclic dependencies ($T\xdep S\xdep T$), indicating non-serializable execution. 

\textbf{Issue 3: Atomicity and Durability.}
A cross-engine transaction should commit either all or none of its sub-transactions. 
Distributed systems usually solve this problem with 2PC, but newer engines may not support it~\cite{EndSlowNetworks,NAMDB}. 
2PC's coordination overhead can also be heavyweight for shared memory, slowing down the (faster) main-memory engine. 
As we describe later, additional checks are needed in addition to a traditional 2PC prepare-commit protocol. 
Thus, 2PC may not be the best choice for single-node multi-engine systems.

\subsection{State-of-the-Art and Motivation}
\label{subsec:motivation}
Prior work can avoid the anomalies~\cite{CO,FDBSSurvey,MDBSOverview,SuperDB,MYRIAD,MYRIADDesign,IncrementalDSI,GeneralizedSI,1SRSI,LazySI}, but they targeted distributed and federated systems without considering single-node fast-slow systems. 
For example, certain solutions for DSI~\cite{IncrementalDSI} guarantee correct snapshots with a central coordinator node and global IDs, requiring non-trivial engine changes. 
2PC as we have described may also not suit fast-slow systems.

\begin{figure}[t]
\centering
\includegraphics[width=\columnwidth]{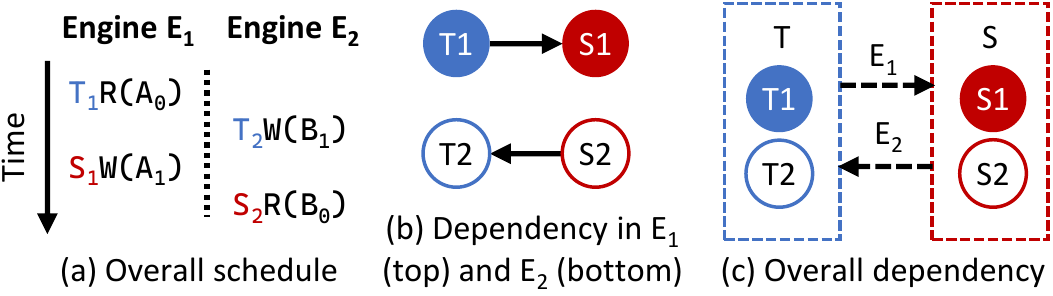}
\caption{\textmd{Non-serializable execution of cross-engine transactions. 
(a) Each engine executes a serializable schedule (b) without cyclic dependencies. 
(c) Overall cyclic dependency between $T$ and $S$.}}
\label{fig:serial-dep}
\end{figure}

The mismatch between prior approaches and fast-slow systems has led to missing or limited cross-engine support in real systems, motivating our work. 
For example, MySQL supports various engines 
under a unified interface. 
Users are free to issue multi-engine transactions, but correctness is undefined as MySQL does not coordinate snapshot and commit ordering {\it across} engines. 
All the anomalies in Section~\ref{subsec:anomalies} could occur. 
Compared to MySQL, SQL Server supports cross-engine transactions with various restrictions~\cite{HekatonWhitepaper}. 
For example, 
if both the traditional engine and Hekaton use SI, cross-engine transactions are not allowed at all, yet SI is among the most popular isolation levels in Hekaton~\cite{HekatonWhitepaper}. 
These significantly limit the potential of cross-engine transactions.

%% file: 3-design.tex
\section{Design Principles}
\label{sec:principles}
We distill a set of desired properties and design principles that a cross-engine mechanism like \skeena should follow:
\begin{itemize}[leftmargin=*]\setlength\itemsep{0em}
\item \textbf{Low Overhead.} 
The mechanism should introduce as low overhead as possible. 
It should try not to penalize single-engine transactions, especially those in the faster engine.  

\item \textbf{Engine Autonomy.} 
Engines should be kept as-is, or only be minimally modified to work with the cross-engine mechanism or optimize for performance. 

\item \textbf{Full Functionality.} The mechanism should support various isolation levels for both single- and cross-engine transactions, unless it is limited by individual engine capabilities. 

\item \textbf{Transparent Adoption.} 
The application should not be required to make logic changes. 
Rather, it should only need to declare the ``home'' engine of each table in the schema. 

\end{itemize}

\section{\skeena Design}
\label{sec:skeena}
\skeena targets fast-slow systems with a main-memory and a storage-centric engine.
We first give an overview of \skeena, and then discuss its design in detail. 

\begin{figure}[t]
\centering
\includegraphics[width=\columnwidth]{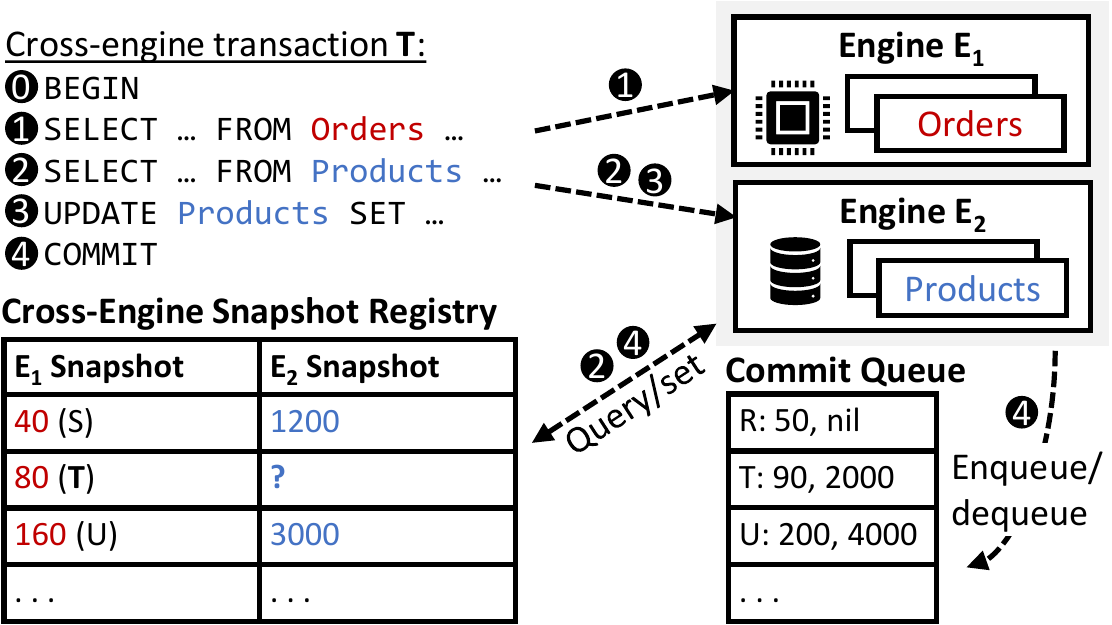}
\caption{\textmd{\skeena overview. 
\protect\circled{0}--\protect\circled{1} Transactions access data without explicitly declaring whether they are cross-engine. 
\protect\circled{2} Upon accessing an additional engine, the transaction \protect\circled{3} consults CSR to obtain a proper snapshot. 
\protect\circled{4} Cross-engine transactions use CSR for commit check and if passed, goes through the pipelined commit protocol.
}}
\label{fig:arch}
\end{figure}

\subsection{Overview}
\skeena ensures correct snapshot selection and atomic commit. 
As Figure~\ref{fig:arch} shows, \skeena consists of 
(1) the cross-engine snapshot registry (CSR) that tracks valid snapshots and 
(2) a pipelined commit protocol for atomically committing cross-engine transactions. 
Now we describe the high-level transaction workflow under \skeena. 

\textbf{Initialization.} 
Transactions (single- or cross-engine) can keep using the database system's unified APIs (e.g., SQL), without 
providing additional hints (e.g., whether a transaction will be cross-engine). 
Figure~\ref{fig:arch} shows an example SQL program, which is written in the same way as without \skeena. 
\skeena does not force transactions to run under specific isolation levels. 
However, the system may allow users to specify an isolation level (e.g., \texttt{SET TRANSACTION ISOLATION LEVEL} in MySQL~\cite{MySQL-SetTX}). 
As part of the integration effort, \skeena can detect and enforce such settings across all engines. 

\textbf{Data Accesses.}
The system routes requests to the target engine which uses a sub-transaction to access data. 
\skeena requires no change to the existing routing mechanism. 
Upon start or accessing the first record, the sub-transaction obtains a snapshot. 
Depending on whether the transaction is single- or cross-engine, the system may directly give the latest snapshot in the underlying engine, or use CSR to obtain a snapshot that would not cause anomalies (steps \circled{2}--\circled{4} \texttt{Query/set}).
If such a snapshot does not exist, the transaction will be aborted; we quantify the impact in Section~\ref{sec:eval}. 

\textbf{Finalization.} 
To commit, a cross-engine transaction consults CSR to verify that committing it would not lead to inconsistent snapshots for future transactions; single-engine transactions commit without using \skeena. 
Transactions that have passed CSR verification are marked as pre-committed to go through the pipelined commit protocol 
(step \circled{4} \texttt{Dequeue}). 
If the verification fails, we abort the transaction by rolling back all of its sub-transactions. 

In the rest of this section, we describe how \skeena facilitates snapshot selection, atomicity and durability, beginning with CSR. 

\begin{algorithm}[t]
\input{algo-get-snap.tex}

\caption{Snapshot selection for cross-engine transactions.}
\label{alg:get-readview}
\end{algorithm}

\subsection{Cross-Engine Snapshot Registry}
\label{subsec:basic-csr}
The key to avoiding inconsistent snapshots is to ensure the sub-transactions of different cross-engine transactions follow the same start order in each engine~\cite{IncrementalDSI}. 
That is, if $T$'s sub-transaction $T_1$ uses an older snapshot than $S_1$ does in engine $E_1$, then $T_2$ should also use an older snapshot compared to that of $S_2$ in $E_2$.
For example, in Figure~\ref{fig:arch}, $T$ first started as a single-engine transaction accessing \texttt{Orders} in $E_1$, using snapshot 80. 
When $T$ starts to access \texttt{Products} in $E_2$, $T$ needs to use a snapshot ($s$) in $E_2$ such that $s$ is between the snapshots of its ``neighbors'' in $E_1$, i.e., $S_1$ and $T_1$.
Thus, $T$ may use any valid $E_2$ snapshot between 1200 and 3000 (inclusive), although using 3000 would allow it to see fresher data. 

To facilitate such a snapshot selection process, CSR tracks valid snapshots (i.e., commit timestamps of past cross-engine transactions) that can be safely used by future cross-engine transactions. 
Conceptually, CSR is a table of many-to-many mappings, where each ``row'' (CSR entry) is a pair of snapshots (i.e., commit timestamps), one from each engine as depicted by Figure~\ref{fig:arch}. 
When a transaction crosses to an additional engine, it uses the current engine's snapshot as the key to query CSR for a snapshot in the other engine. 
As Algorithm~\ref{alg:get-readview} shows, to access a new engine \texttt{e2}, the worker thread issues a non-inclusive forward scan over CSR using the snapshot in the current engine \texttt{e1} as the key (\texttt{e1\_snap}) to obtain a set of candidate snapshots. 
The scan returns once a first key greater than \texttt{e1\_snap} is met or no such key is found. 
If the scan returns an empty set, then no past transaction has set up any mapping, or the current transaction is using the latest \texttt{e1} snapshot. 
Then we use the latest \texttt{e2} snapshot (lines 4--6). 
However, if any candidate is found, we must take an \texttt{e2} snapshot that is already mapped to \texttt{e1\_snap} to avoid anomalies (lines 7--9). 
Finally, we ensure the mapping is recorded at line 10. 
Under SI, the algorithm is executed only once per transaction when it becomes cross-engine.
Subsequent accesses continue to use the previously acquired snapshots.

\begin{algorithm}[t]
	\input{algo-commit-check.tex}
	\caption{Commit check for cross-engine transactions.}
	\label{alg:commit-check}
\end{algorithm}

In addition to acquiring snapshots, committing a cross-engine transaction implicitly limits the ranges of snapshots a (future) cross-engine transaction may use: recall that the commit timestamp of a previous transaction $T$ in fact is the snapshot of a future transaction that reads the results generated by $T$.
Thus, CSR also tracks commit timestamps of cross-engine transactions. 
Algorithm~\ref{alg:commit-check} describes the process at a high level. 
Here, we assume the sub-transaction commit timestamps are given (as the \texttt{commit\_ts} member in each sub-transaction); we revisit this assumption later in more detail. 
The idea is to ensure that committing a cross-engine transaction---i.e., adding a new mapping entry to CSR---would not add skewed snapshots to CSR. 
Thus, upon commit, we issue a reverse scan and a forward scan over CSR using a sub-transaction's (\texttt{sub\_t1}) commit timestamp to obtain the lower and higher bounds for the other commit timestamp (lines 4--11 in Algorithm~\ref{alg:commit-check}). 
If \texttt{sub\_t2}'s commit timestamp falls between the higher and lower bounds, we can safely commit this cross-engine transaction and setup a new mapping in CSR (line 18).
Otherwise, the transaction must be aborted. 
Note that the mapping step in Algorithm~\ref{alg:get-readview} is still necessary: (1) single-engine commits are not covered by CSR to avoid unnecessary overheads, and (2) a cross-engine transaction may access data generated by single-engine transactions and form new cross-engine snapshots. 

Since a transaction may access engines in any order (e.g., from the storage-centric engine and crosses over to the memory-optimized engine, and vice versa), CSR needs to support queries from \textit{either} engine. 
CSR may be implemented using a relational table in one of the supported engines with full-table scan or two range indexes, each of which is built on a ``column'' of the CSR table. 
However, this can create dependency on a particular engine and incur much table and index maintenance overhead. 
A practical design must also support concurrency. 
We address these issues next. 

\subsection{Lightweight Multi-Index CSR}
\label{subsec:multi-index}

We take advantage of the unique properties of fast-slow systems to devise a lightweight CSR that mitigates the above issues. 

\textbf{Anchor Engine.}
Compared to storage-centric engines, it is typically much cheaper to obtain snapshots in main-memory engines. 
This is often as simple as manipulating an 8-byte counter in a lock-free manner without using a mutex~\cite{Abyss}. 
For example, ERMIA~\cite{ERMIA} only needs to read the counter; 
Hekaton~\cite{Hekaton} increments the counter using atomic fetch-and-add (\FAA)~\cite{IntelManual} to keep the process efficient.\footnote{Such designs are common in multi-versioned engines~\cite{MVCCEval,Abyss,Bohm,Cicada}. 
When the engine is integrated into a full system, e.g., MySQL , \FAA's overhead is negligible.} 
On the contrary, obtaining a snapshot in a storage-centric system can be much more complex. 
For example, MySQL InnoDB needs to take multiple mutexes to compute watermark values (see Section~\ref{sec:mysql}). 

Leveraging the existence of a fast and a slow engine, \skeena designates an \textit{anchor} engine and always follows the snapshot order in the anchor engine. 
The anchor engine should be the one where it is cheaper to acquire a snapshot (usually the memory-optimized one). 
Then a transaction always starts by acquiring the latest snapshot from the anchor engine, and uses it to query CSR when it extends to the other engine. 
This allows us to maintain one-to-many mappings (instead of many-to-many mappings), which simplifies CSR to become a range index that uses the anchor engine's snapshots as ``keys'' and lists of snapshots in the other engine as ``values.'' 
We currently use Masstree~\cite{Masstree}, a high-performance in-memory index, but any concurrent data structure that supports range queries would suffice.
A side effect is transactions that only access the ``slower'' engine become cross-engine and need to use Algorithms~\ref{alg:get-readview}--\ref{alg:commit-check}. 
As Section~\ref{sec:eval} shows, the overhead is negligible compared to data accesses which may involve the storage stack while CSR is fully in-memory. 

Using the main-memory engine as the anchor is an optimization, not a requirement: 
in theory any engine can be the anchor. 
In case a heavyweight engine has to be the anchor, cross-engine transactions may incur higher overhead for creating snapshots (thus lower overall performance). 
This will in turn reduce the pressure on CSR which is less frequently accessed and maintains fewer snapshots. 

\textbf{Multi-Index.}
Since CSR tracks cross-engine snapshots and commit histories, its size can grow quickly, slowing down query speed over time; 
entries that are no longer needed should also be cleaned up. 
We solve these problems by partitioning the CSR by snapshot ranges, reminiscent of multi-rooted B-trees~\cite{PLP}. 
The result is a multi-index design shown in Figure~\ref{fig:multi-index}.
Each partition is an index and covers a unique range of snapshots so that a transaction only uses a single index. 
In Figure~\ref{fig:multi-index}, the first two indexes cover mappings in the ranges of [30, 400] and [401, 500], respectively. 
Each partition has a fixed capacity (number of keys), and a new index is created when the current open index is full. 
Therefore, there is always one and only one open index that can accept new mappings; 
other indexes are read-only but can continue to serve existing transactions for snapshot selection. 
However, since inactive indexes are read-only, if a transaction needs to setup a new mapping in an inactive index during snapshot selection or commit check, it must be aborted; in practice, such aborts are rare as we evaluate in Section~\ref{sec:eval}. 

\begin{figure}[t]
\centering
\includegraphics[width=0.95\columnwidth]{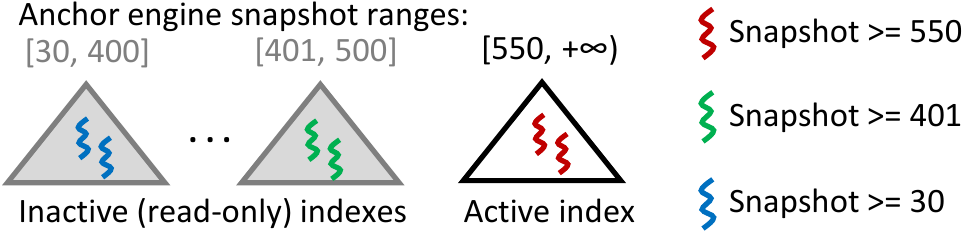}
\caption{\textmd{Multi-index CSR. 
Each index covers a range of anchor snapshots, and is recycled in its entirety when no longer needed.}}
\label{fig:multi-index}
\end{figure}

\textbf{Snapshot Acquisition and Commit Check.}
With multiple indexes and an anchor engine, a transaction acquires snapshots by 
(1) obtaining a snapshot $S$ from the anchor engine, 
(2) locating the index $I$ that covers $S$, and 
(3) using $S$ to query $I$ and if needed, create a new entry in $I$ following Algorithm~\ref{alg:get-readview} with \texttt{e1\_snap = $S$} and \texttt{CSR} at line 3 being $I$. 
Note that steps 2 and 3 are only executed if the transaction accesses the non-anchor engine. 
For example, if the main-memory engine is the anchor and the transaction only accesses an in-memory table, steps 2 and 3 are never executed. 
For step 2, we track all the indexes in a list/array (e.g., C++ \texttt{std::vector}). 
Each entry records the minimum snapshot of the partition and a pointer to the index. 
Since we keep only one open index, entries in the list are sorted by snapshot ranges. 
We search for $I$ by traversing the list backwards and stopping at the first entry whose smallest snapshot is smaller than or equal to the given snapshot. 
In step 3, a new entry is inserted if and only if $I$ is open; 
otherwise the transaction is aborted. 
Commit check follows a similar logic and can proceed only if $I$ is open. 
Likewise, \texttt{CSR} in lines 4--9 of Algorithm~\ref{alg:commit-check} refers to $I$.
Multiple threads may execute the above three steps concurrently, for which we describe our solution next.

\subsection{CSR Concurrency and Maintenance}
\label{subsec:csr-concurrency}
Now we discuss how \skeena handles concurrent accesses and manages/recycles indexes in multi-index CSR.

\textbf{Concurrency.}
Although latches can be a potential bottleneck in multicore systems~\cite{Shore-MT}, a latch-based solution in \skeena can be efficient by leveraging the fast-slow property: 
compared to executing transactions in the slower engine, using latches and high-performance indexes present negligible overhead and little impact on overall performance for cross-engine transactions.
Each index is protected by a mutex, and we protect the array of all indexes using a reader-writer lock for mutual exclusion between threads that only query an index without modifying the list using the reader mode and those that may add or remove an index using the writer mode. 
A transaction starts by latching the list of indexes in shared (reader) mode to locate the target index $I$. 
Then the thread latches $I$ for exclusive access to run Algorithm~\ref{alg:get-readview}. 
If the thread needs to create a new index, e.g., if $I$ is full or the list is empty, it (1) releases the list latch, (2) requires the latch in writer mode to allow inserting to the list, and (3) checks if such an index has been inserted by another thread between steps 1 and 2, and if so, retries the entire process after releasing the list latch; 
otherwise we proceed by (4) appending the new index with a new mapping to the list. 
Commit check follows the same logic so we do not repeat the details.

\textbf{Index Maintenance.}
Using multiple indexes simplifies garbage collection (GC) as we can delete an entire index once its mappings are no longer needed, instead of issuing many key delete operations to an index. 
To recycle, we first iterate over all the active transactions to find the oldest anchor-engine snapshot (\texttt{min\_snap}). 
Then we exclusively latch the list of indexes and scan through it to remove stale indexes that cover ranges below \texttt{min\_snap}.
In case long-running transactions prevent \texttt{min\_snap} from growing, one may further find opportunities to remove unused indexes covering newer but still unused ones, reminiscent of GCing long version chains in multi-versioned systems~\cite{Steam}. 
For example, assume a transaction still uses the left-most index in Figure~\ref{fig:multi-index}, yet no transactions need the middle index (401--500), which can then be first recycled; 
the oldest index is recycled later when it is no longer needed. 
Our experiments do not indicate this approach to be necessary and recycling is fast as it is purely in-memory. 
Recycling is triggered between CSR accesses based on a user-defined threshold (e.g., once per 5000 accesses); 
it could also be delegated to background threads.

\subsection{Commit Protocol}
\label{subsec:commit}
Once all accesses are finished, \skeena checks whether both sub-transactions can commit using engine-level commit timestamps that represent the sub-transactions' commit ordering. 
Thus, \skeena needs to obtain the sub-transaction's commit timestamp from each engine. 
This is usually easy for memory-optimized engines which break the commit process into pre- and post-commit~\cite{ERMIA,HekatonCC,Hekaton}. 
During pre-commit, the engine assigns a commit timestamp and uses it to determine whether the transaction can commit without violating correctness criteria (e.g., serializability). 
If so, post-commit will finish the commit by marking new records as visible, finalizing log records, etc.
Otherwise the transaction is aborted. 
Some (mainly storage-centric) engines may not explicitly expose such pre- and post-commit interfaces, 
but engines in a multi-engine system are maintained by the same vendor. 
This justifies simple changes in engines to expose the pre- and post-commit interfaces, which is straightforward in practice by breaking a monolithic ``commit'' function into a pre- and post-commit function (Section~\ref{sec:mysql}). 
Single-engine transactions directly execute the two steps without commit check. 

With the pre- and post-commit interfaces, \skeena commits a cross-engine transaction in three steps. 
(1) Pre-commit both sub-transactions to obtain commit timestamps. 
(2) Use the timestamp obtained from the anchor engine to conduct the commit check. 
(3) If the check passes, post-commit both sub-transactions.
From a high-level, \skeena's commit protocol resembles 2PC:
step 1 may correspond to 2PC's prepare phase that collects commit decisions from each engine;
step 3 may correspond to 2PC's commit phase. 
However, \skeena differs from 2PC by requiring an additional check (step 2) even after all the engines have pre-committed the transaction. 
So an ``all-yes'' result from the 2PC-equivalent prepare phase does not necessarily mean a cross-engine transaction can commit.

Before both sub-transactions are post-committed, changes by either should be kept invisible, yet 
from the perspective of an engine, a post-committed (sub-)transaction is fully committed with its results visible. 
\skeena must ensure partial results are not visible until all sub-transactions are post-committed.
We observe that a simple yet effective solution is to extend the pipelined commit protocol~\cite{Aether} which was initially proposed to hide log flush latency. 
It decouples transactions waiting for log flushes and worker threads to keep I/O off the critical path. 
Upon commit, instead of directly issuing a log flush, the thread detaches the transaction and appends it to a global commit queue (or a partitioned queue to avoid introducing a central bottleneck). 
Results by these transactions are immediately visible internally but are not returned to applications until their log records have been persisted. 
A daemon tracks transactions awaiting log durability on the commit queue, and dequeues transactions whose log records have been persisted. 
Some systems use this approach to improve throughput without sacrificing correctness~\cite{Aether,QueryFresh,NVM-DLog,Shore-MT}. 

Based on this idea, \skeena (1) pushes both sub-transactions onto the commit queue upon post-commit and (2) has the commit daemon monitor both engines' log flushes to dequeue transactions. 
If an engine already implements commit pipelining, \skeena can directly extend it. 
Note that single-engine and read-only transactions must also use commit pipelining~\cite{CLV} as they may read cross-engine transactions' results; we quantify its impact on latency in Section~\ref{sec:eval}.

\subsection{Durability and Recovery}
In a multi-engine system, each engine implements its own approach to durability and crash recovery.
Sub-transactions still follow their corresponding engines' approach to persist data and log records. 
Checkpoints can be taken as usual independently by each engine. 
To ensure atomicity of cross-engine transactions, \skeena can record the pre- and post-commit of cross-engine transactions, by maintaining a standalone log or piggybacking on individual engines. 
The latter can be easier to implement: upon pre-commit we append a \texttt{commit-begin} record, and after post-commit finishes, the engine appends a \texttt{commit-end} record. 
During recovery, each engine executes its recovery mechanism and rolls back changes done by cross-engine transactions whose sub-transactions are not fully committed. 
Alternatively, the recovery procedure may inspect each engine's log and truncate at the first ``hole'' where only one sub-transaction of a cross-engine transaction is committed. 
This is safe because transactions that depend on partially committed cross-engine transactions will wait on the commit queue and their results were never made visible to applications. 

\subsection{Serializability}
\label{subsec:serial}
As noted by prior work~\cite{CO}, disallowing anti-dependencies (i.e., using commit order as dependency order) in all engines is sufficient for cross-engine serializability. 
This translates into choosing a concurrency control protocol for each engine where a sub-transaction can only commit if its read records are not concurrently modified by a newer transaction. 
A wide range of engines~\cite{Hekaton,HekatonCC,Bohm,Silo,FOEDUS,MOT,Cicada,Deuteronomy} exhibit this property based on 2PL (by blocking readers and writers) and OCC (by verification at commit time). 
Some protocols~\cite{SSN,SGT,MakingSSI,SSI} can tolerate certain safe anti-dependencies, but would require implementing verification in \skeena. 
This needs engines to expose dependency information, tightly coupling \skeena with engine design and sacrificing engine autonomy. 
Thus, we take the former approach that imposes no engine-level changes. 

\subsection{Correctness}
\label{subsec:proof}
\skeena's theoretical foundation comes from DSI~\cite{IncrementalDSI} and commit ordering (CO)~\cite{CO} which respectively ensure consistent snapshots and serializability. 
Different from DSI and CO, \skeena targets single-node instead of distributed systems. 
In essence, \skeena implements DSI and CO for shared memory. 
Both DSI and CO enforce sufficient conditions with their correctness formally proved. 
Thus, we argue for \skeena's correctness by showing \skeena enforces the same conditions as DSI and CO. 
We first lay out the necessary notations used by DSI~\cite{IncrementalDSI} and our adaptation for shared memory:

\begin{itemize}[leftmargin=*]\setlength\itemsep{0em}
\item $c_x$: Commit of transaction $x$; 
\item $c_x^i$: Commit of transaction $x$ on node/engine $i$; 
\item $b_x$: Begin of transaction $x$; 
\item $b_x^i$: Begin of transaction $x$ on node/engine $i$; 
\item $SN(x, y)$: All nodes/engines accessed by $x$ and $y$;
\item $x < y$: Transaction $x$ is serialized before $y$; 
\item $op_1 < op_2$: Begin/commit operation $op_1$ is ordered before $op_2$. 
For begin operations, $op_1 \le op_2$ is allowed as begin timestamps can be acquired by reading a counter. 
\end{itemize}
Then DSI defines the following sufficient conditions for correctness: 

\begin{theorem}\label{the:dsi}
Suppose each node $i$ enforces correct local SI. 
If all the local schedules satisfy the following rules, we can construct a correct DSI schedule, i.e., the execution is correct DSI~\cite{IncrementalDSI}: 

$\exists c_x^i < c_y^i \xdep c_x < c_y$ \hfill (Rule 1)  

$\exists b_x^i < c_y^i \xdep b_x < c_y$ \hfill (Rule 2) 

$\exists c_x^i < b_y^i \xdep c_x < b_y$ \hfill (Rule 3) 

$\exists b_x^i \le b_y^i \xdep b_x \le b_y$ \hfill (Rule 4) 

$c_x < c_y \xdep \forall j \in SN(x, y): c_x^j < c_y^j$ \hfill (Rule 5) 

$b_x < c_y \xdep \forall j \in SN(x, y): b_x^j < c_y^j$ \hfill (Rule 6) 

$c_x < b_y \xdep \forall j \in SN(x, y): c_x^j < b_y^j$ \hfill (Rule 7) 

$b_x < b_y \xdep \forall j \in SN(x, y): b_x^j < b_y^j$ \hfill (Rule 8) 
\end{theorem}

Analogous to DSI, in a shared-memory environment, Rules 1--4 construct the partial order of begin/commit events of \textit{cross-engine} transactions (``distributed transactions'' in DSI). 
Then, Rules 5--8 enforce the same partial order across all the \textit{engines} (``nodes'' in DSI). 
We omit the detailed proof (available elsewhere~\cite{IncrementalDSI}); our goal is to show the conditions enforced by \skeena satisfy Rules 1--8.

\begin{theorem}\label{the:skeena}
\skeena enforces the same partial order of cross-engine transactions on all engines as defined by Theorem~\ref{the:dsi} in a shared-memory environment. 
\end{theorem}

\begin{proof}
In Algorithm~\ref{alg:get-readview}, given some $b_x^i$, \skeena either uses the latest timestamp (line 6) or the latest snapshot from engine $j$ whose corresponding snapshot in $i$ is no newer than $b_x^i$. 
Therefore, Algorithm~\ref{alg:get-readview} enforces Rules 4 and 8 of Theorem~\ref{the:dsi}. 
Similarly, Algorithm~\ref{alg:commit-check} forbids $<c_x^i, c_y^j>$ to be inserted into CSR if $c_y^j$ does not maintain the same partial order on engine $i$ (lines 14--15), satisfying Rules 1 and 5. 
Note that begin timestamps are in fact previously committed transactions' commit timestamps, enforcing Rules 3, 4, 6 and 7. 
\end{proof}

For serializability, \skeena requires the use of concurrency control protocols to follow the requirements of CO~\cite{CO}, without introducing any additional algorithms. 
Hence, \skeena can support serializability using commit ordering. 

\subsection{Discussions}
\label{subsec:discuss}
In essence, \skeena is a coordinator that enforces correct snapshots and atomic commit in fast-slow systems. 
Table~\ref{tbl:isolation} lists the requirements to achieve different overall isolation levels. 
To achieve an overall isolation level of read committed (RC), in most systems this means to acquire a new snapshot per record access. 
Guaranteeing SI or serializable isolation levels usually requires obtaining a snapshot upon transaction start or the first record access. 
Since \skeena does not implement extra concurrency control logic to avoid tight coupling with engines, for all isolation levels (e.g., SI), each engine needs to run at least at it (e.g., SI) or higher to ensure sub-transactions are correctly scheduled.
Thus, the overall isolation level guaranteed by \skeena is at most the lower level being used across all engines. 
For example, if two engines respectively use RC and SI, then \skeena can only guarantee RC overall. 

\begin{table}[t]
\caption{\textmd{Requirements to achieve different isolation levels.}}
\label{tbl:isolation}
\begin{tabular}{@{}p{2.4cm}p{5.7cm}@{}}
\toprule
\textbf{Isolation Level} & \textbf{Requirements}  \\ \midrule
Read Committed & \textit{Engines:} Isolation level $\ge$ Read Committed\\
& \textit{\skeena}: Refresh snapshot per record access\\ \midrule
Snapshot Isolation& \textit{Engines:} Isolation level $\ge$ Snapshot Isolation\\
& \textit{\skeena}: New snapshot upon start/first access\\\midrule
{Serializable} & \textit{Engines:}  Serializable, no anti-dependencies\\
& \textit{\skeena}: New snapshot upon start/first access\\
\bottomrule
\end{tabular}
\end{table}

Our focus has been on dual-engine systems with interpreted queries. 
To support more engines, a straightforward way is to extend the multi-index CSR to become a hierarchy of indexes, each is the anchor of the next lower-level engine. 
This would allow \skeena to enforce ordering between sub-transactions. 
The downside is more complex CSR maintenance which may require more efficient CSR designs. 
If more changes are tolerable, one may introduce additional global ordering to simplify sub-transaction ordering. 
However, this can potentially couple the design with engine internals. 
Most memory-optimized engines compile queries to machine code~\cite{Hekaton,PermutableCompiledQueries,HyperCompile}. 
\skeena is orthogonal to whether queries are compiled or interpreted, although we focus on the latter as accesses in conventional engines can cancel out compilation's benefits. 

Finally, \skeena can be applied to systems that (1) support multiple engines and (2) follow the database model in Section~\ref{subsec:model}.
Both are widely available in practice. 
For example, MySQL and SQL Server already support multiple engines. 
PostgreSQL can support additional engines using foreign data wrapper. 
Many systems, including MySQL, PostgreSQL and SQL Server, employ multi-versioning. 
We discuss in depth how \skeena can be used by MySQL later; for space limitation we do not expand on other systems. 
Moreover, \skeena does not require significant engine-level changes. 
The most notable (yet simple) change (mainly for conventional engines) is exposing commit ordering via a pre-commit interface. 
\skeena only expects the commit/abort decision of sub-transactions, 
without dictating engine internals, such as whether cascading abort is possible or how writes and versions are organized.  
The remaining effort is mainly put into integrating \skeena with existing multi-engine support. 
As Section~\ref{sec:mysql} describes, these changes are not intrusive or complex.

%% file: algo-get-snap.tex
\begin{lstlisting}[language=python, mathescape, gobble=0, keywordstyle=\ttfamily\bfseries\color{blue},]
def select_snapshot(e1_snap, engine &e2):
  # Find existing snapshots that could be used
  candidates[] = CSR.forward_scan_1st(e1_snap)
  if candidates is empty:
    # No existing mapping, obtain the latest from e2
    e2_snap = e2.timestamp_counter
  else:
    # Use the latest snapshot mapped to s <= e1_snap
    e2_snap = max(candidates)
  CSR.map(e1_snap, e2_snap)
  return e2_snap 
\end{lstlisting}

%% file: algo-commit-check.tex
\begin{lstlisting}[language=python, gobble=0, mathescape, keywordstyle=\ttfamily\bfseries\color{blue},]
def cross_engine_commit_check(sub_t1&, sub_t2&): 
  # Obtain lower and upper bounds for sub-transaction t2 
  low = -inf
  candidates[] = CSR.reverse_scan_1st(sub_t1.commit_ts)
  if candidates[] is not empty:
    low = max(candidates)

  high = +inf
  candidates[] = CSR.forward_scan_1st(sub_t1.commit_ts)
  if candidates[] is not empty:
    high = min(candidates)

  # Check if committing t2 would cause future anomalies
  if low > sub_t2.commit_ts or high < sub_t2.commit_ts:
    return false
  else:
    # Check passed, setup mapping and return
    CSR.map(sub_t1.commit_ts, sub_t2.commit_ts)
    return true
\end{lstlisting}

%% file: 4-mysql.tex
\section{Skeena in Practice}
\label{sec:mysql}
We explore the effort needed to adopt \skeena in real systems, by
enabling cross-engine transactions in open-source MySQL\footnote{Based on MySQL 8.0 at \url{https://github.com/mysql/mysql-server}.} between its default storage-centric InnoDB and ERMIA~\cite{ERMIA}.
MySQL defines a set of core interfaces (e.g., search, update and commit) for engines to implement~\cite{MySQLAltEngines}.
This allowed us to integrate ERMIA easily with $<2000$ LoC.\footnote{Details in \texttt{ha\_ermia.cc} in our code repository (\url{https://github.com/sfu-dis/skeena}).}
InnoDB and ERMIA share MySQL's SQL layer and thread pool.\footnote{Adopted from \url{https://github.com/percona/percona-server/blob/8.0/sql/threadpool.*}.}
The application specifies each table's home engine in its schema, which is managed by existing MySQL features.

To use CSR, it is necessary to understand each engine's database model.
ERMIA closely follows our database model.
To obtain a snapshot, the thread reads the counter without latching, which is much cheaper than InnoDB based on latching.
So we use ERMIA as the anchor engine.
InnoDB uses transaction IDs (TIDs) to determine record visibility and ordering.
We describe how we reconcile the differences between real implementation and our database model.
Each read-write transaction is uniquely identified by a TID drawn from a central counter.
Each record is stamped with the TID of the transaction that last updated it.
Updates are handled in-place.
Old versions are generated on-demand using undo logs.
The freshness (or the amount of undo log to apply) is determined by the transaction's read view (snapshot) which is acquired upon the first data access.
A read view consists of low/high watermarks (TIDs) and an active transactions list captured at the read view's creation time.
The transaction is not allowed to see versions created by transactions with TIDs above the high watermark, but can see the results of transactions with TIDs below the low watermark.
Versions created by transactions with TIDs between the two watermarks are invisible if they are active.
As a result, read views are not directly comparable, deviating from our database model.

Our solution is to use the high watermark in CSR.\footnote{Lines 2233--2254 of \texttt{trx0trx.cc} in our code repository.}
Specifically, sub-transactions in InnoDB first acquire the latest read view using the original approach.
We then adjust its high watermark using CSR and leave the active transactions list unchanged.
In case the new high watermark is even lower than the low watermark, we adjust both to be the same as the high watermark.
The sub-transaction can then use the adjusted read view as usual to test record visibility.

Upon commit, InnoDB assigns a \texttt{serialisation\_no} drawn from the TID counter to denote commit ordering, which we use for \skeena's commit check.\footnote{Lines 1378--1480 of \texttt{trx0trx.cc} in our code repository.}
We broke the monolithic commit function into pre/post-commit functions (Section~\ref{subsec:commit}).
The pre-commit function only acquires a \texttt{serialisation\_no}, leaving the remaining logic to post-commit.
For atomic commit, we piggyback on ERMIA's commit pipelining.
We extend the commit entry design in ERMIA to include commit LSNs in both engines, along with a MySQL callback for notifying the client of concluded transactions.\footnote{Lines 180--222 of \texttt{sm-log-alloc.cpp} in our code repository.}

In total, we modified 83 LoC in InnoDB for it to use \skeena to choose read views and commit sub-transactions.
CSR is implemented as a separate module of $\sim$600LoC.\footnote{Details in \texttt{gtt.\{cc,h\}} in our code repository.}
For ERMIA, we only modified its commit pipelining code to consider both engines.

%% file: 5-evaluation.tex
\section{\skeena in Action}
\label{sec:eval}
We empirically evaluate \skeena under microbenchmarks and realistic workloads. 
Through experiments, we show that: 

\begin{itemize}[leftmargin=*]\setlength\itemsep{0em}
\item \skeena retains the performance benefits brought by memory-optimized engines in fast-slow systems; 
\item \skeena only incurs a very small amount of overhead for cross-engine transactions; 
\item By judiciously placing tables in different engines, \skeena can effectively improve performance for realistic workloads.
\end{itemize}

\subsection{Experimental Setup}
We run experiments on a dual-socket server equipped with two 20-core Intel Xeon Gold 6242R CPUs (80 hyperthreads in total), 384GB of main memory \edit{and a 400GB Micron SSD with peak bandwidth of 760MB/s.} Each CPU has 35.75MB of cache and is clocked at 3.1GHz.  
All experiments are conducted in MySQL 8.0 with InnoDB and ERMIA. 
We use SysBench~\cite{SysBench} to issue benchmarks. 
To reduce networking overhead, we pin MySQL server and the client (SysBench) to two different CPU sockets, and  
use a Unix Domain Socket between the server and client~\cite{MySQLTCPSocket}. 
We use jemalloc~\cite{jemalloc} to avoid memory management becoming a major bottleneck. 
We use SI (repeatable read in InnoDB) to run all experiments as they are widely used in practice, and reinitialize the database for each run which then starts with a warm buffer pool. 
We report the average throughput and latency of three 60-second runs. 

ERMIA is memory-optimized so all records are in heap memory. 
For InnoDB, we test both the memory- and storage-resident cases: 
the memory-resident variant (\texttt{InnoDB-M}) uses a large enough buffer pool to avoid accessing storage; 
the storage-resident variant (\texttt{InnoDB}) uses a small buffer pool that would mandate accessing the storage stack. 
\edit{To stress test \skeena, we store persistent data (such as data files and logs) in \tmpfs, so that I/O is as fast as memory, making it easier to expose \skeena's overhead. 
To understand the performance under more realistic workloads, we also run experiments using a real SSD; 
\tmpfs is used unless otherwise specified.}

\subsection{Benchmarks}
We use YCSB-like~\cite{YCSB} microbenchmarks and TPC-C~\cite{TPCC} (based on Percona's implementation~\cite{PerconaTPCC}) to test \skeena and explore the effect of cross-engine transactions. 

\textbf{Microbenchmarks.} 
We devise three microbenchmarks based on access patterns: read-only, read-write and write-only. 
Unless otherwise specified, each transaction accesses ten records randomly chosen from a set of tables following a uniform distribution: 
(1) for read-write transactions, eight out of the ten accesses are point reads and two are updates, and 
(2) for each engine, we create 250 tables, each of which contains a certain number of records depending on whether the experiment is memory- or storage-resident for InnoDB.
Each record is 232-byte, consisting of two \texttt{INTEGER}s and one \texttt{VARCHAR}. 
For memory-resident experiments, each table contains 25000 records, bringing the total data size of 250 tables to $\sim$1.35GB; the buffer pool size in InnoDB is set to 32GB. 
For storage-resident experiments, we set each table to contain 250000 records, and the total data size is $\sim$13.5GB; we set the buffer pool to be 2GB. 
Under both settings, ERMIA is populated with the same amount of data as InnoDB (i.e., 500 tables across two engines).

\begin{table}[t]
\centering
\caption{\textmd{Throughput (TPS) of single-engine microbenchmarks (80 connections) and TPC-C (50 connections). 
\skeena (-S) incurs negligible overhead and retains ERMIA's high performance.}}
\label{tbl:single-engine}
\begin{tabular}{@{}lllll@{}}
\toprule
\bf Scheme    & \bf Read-only & \bf Read-write & \bf Write-only & \bf TPC-C   \\ \midrule
\ermia & 1,427,071 & 1,252,146 & 1,091,606 & 7,550 \\
\ermias & 1,430,137 & 1,253,368 & 1,095,056 & 7,546 \\ \hline
\innodbm & 1,326,710 & 930,249 & 710,697 & 626 \\ 
\innodbms & 1,310,809 & 915,406 & 711,425 & 612 \\ \hline
\innodb & 456,672 & 420,328 & 194,446 & 277 \\
\innodbs & 453,781 & 420,474 & 194,412 & 261 \\ \bottomrule
\end{tabular} 
\end{table}

\textbf{TPC-C}. 
We use TPC-C for the dual-purpose of (1) testing \skeena under non-trivial transactions, and (2) exploring the potential benefits of cross-engine transactions in realistic scenarios. 
We run both memory- and storage-resident experiments: the former sets the scale factor to be the number of connections and the latter uses 200 warehouses. 
For memory-resident experiments, each connection works on a different home warehouse, but the 1\% of New-Order and 15\% of Payment transactions may respectively access a remote warehouse; we set InnoDB buffer pool to be 32GB which is large enough to hold all the data ($\sim$14GB). 
With 200 warehouses for storage-resident experiments, the total data size is $\sim$55GB, for which we set InnoDB to use a buffer pool of 5GB and set each thread (connection) to always pick a random warehouse as its home warehouse to ensure the footprint covers the entire database.
Finally, we gradually move tables from InnoDB to ERMIA, making the affected transactions cross-engine. 
This allows us to explore the effectiveness of cross-engine transactions and distill several useful suggestions on how to optimize performance in fast-slow systems;
we discuss more detailed setups later.

\subsection{Single-Engine Performance}
\label{subsec:single-perf}
An important goal of \skeena is to ensure single-engine transactions (especially those in the faster engine, ERMIA) pay little additional cost. 
We evaluate this aspect by turning \skeena on and off under six ERMIA- and InnoDB-only variants. 
\edit{To stress test \skeena, we use the memory-resident InnoDB (\texttt{InnoDB-M}) and the storage-resident InnoDB with \tmpfs (\texttt{InnoDB}).} 
Table~\ref{tbl:single-engine} summarizes the results; variants with \skeena turned on carry an \texttt{S} suffix. 
In all cases, \skeena incurs negligible overhead with the slightly more complex logic in commit pipelining.
Note that ``single-engine'' transactions in \innodb/\innodbm are in fact cross-engine, as they must follow the start order in the anchor engine (ERMIA) even if they do not access any records in ERMIA. 
This means CSR will only maintain a single mapping (using ERMIA's initial snapshot) which incurs a constant but very small amount of overhead (up to 5.6\%); garbage collection is also never needed with a single mapping.  
Compared to \innodb, \innodbm performs up to over $\sim3\times$ better thanks to its large buffer pool. 
\innodbm and \innodbms perform similarly to ERMIA under the read-only microbenchmark, but fall behind as we add more writes to the workload, signifying the potential benefits a memory-optimized engine could bring (more later). 
\ermias performs as well as \ermia since CSR is never used. 
These results verify that \skeena retains the advantage of memory-optimized engines. 

\subsection{Cross-Engine Performance}
\label{subsec:cross-perf}
Now we explore the behavior of cross-engine transactions using microbenchmarks. 
For each transaction, we vary the percentage of InnoDB and ERMIA accesses out of ten accesses. 
For example, with \texttt{30\% InnoDB}, three accesses per transaction are done in InnoDB, the remaining seven accesses go to ERMIA.
We use the same \texttt{-M} and \texttt{-S} notations from Section~\ref{subsec:single-perf} for single-engine transactions whose results are shown to calibrate expectations. 
For cross-engine transactions we mark the percentage of InnoDB accesses and note whether the experiment is memory- or storage-resident as needed.

\begin{figure}[t]
\centering
\includegraphics[width=\columnwidth]{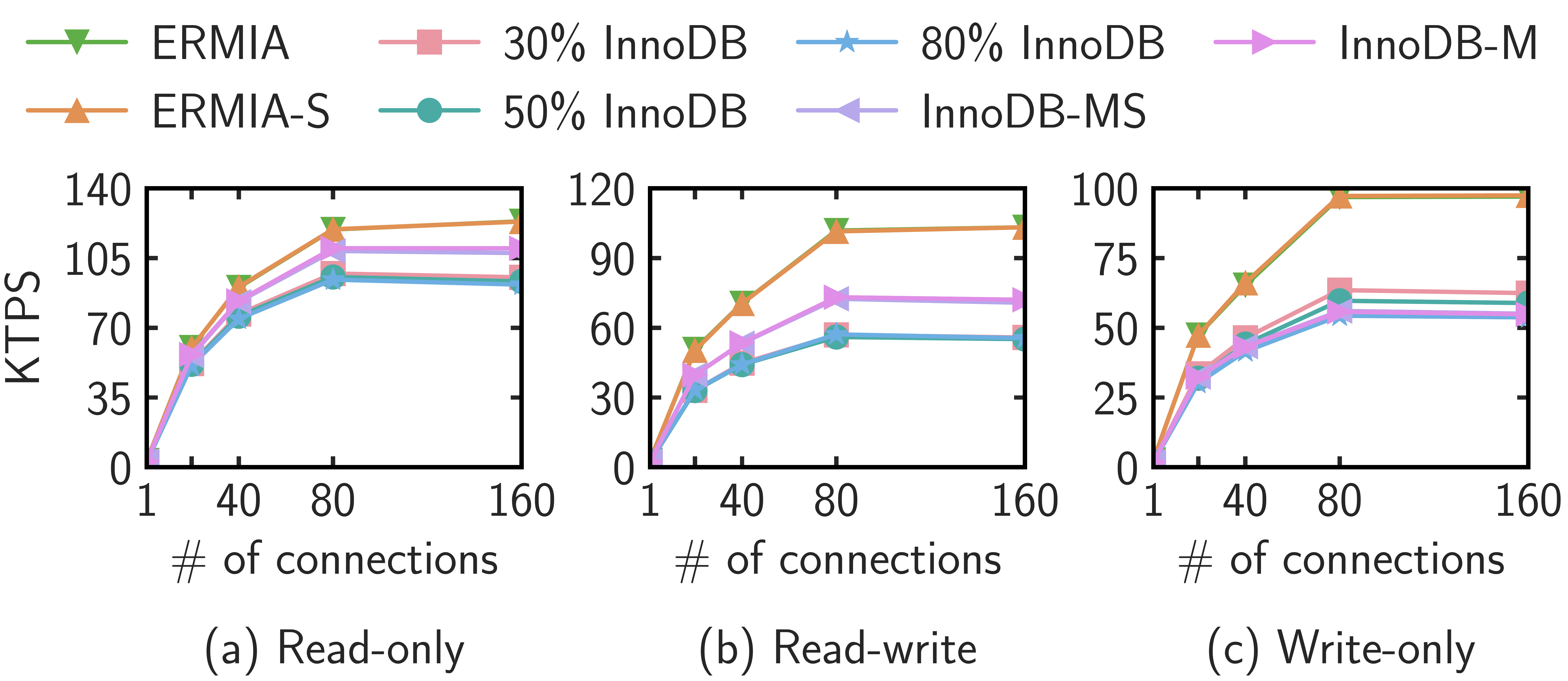}
\caption{\textmd{Throughput under memory-resident microbenchmarks. 
CSR cost can be comparable to that of reading records, causing \innodbm to outperform cross-engine cases.}}
\label{fig:microbenchmark-memory}
\end{figure}

\begin{figure}[t]
\centering
\includegraphics[width=\columnwidth]{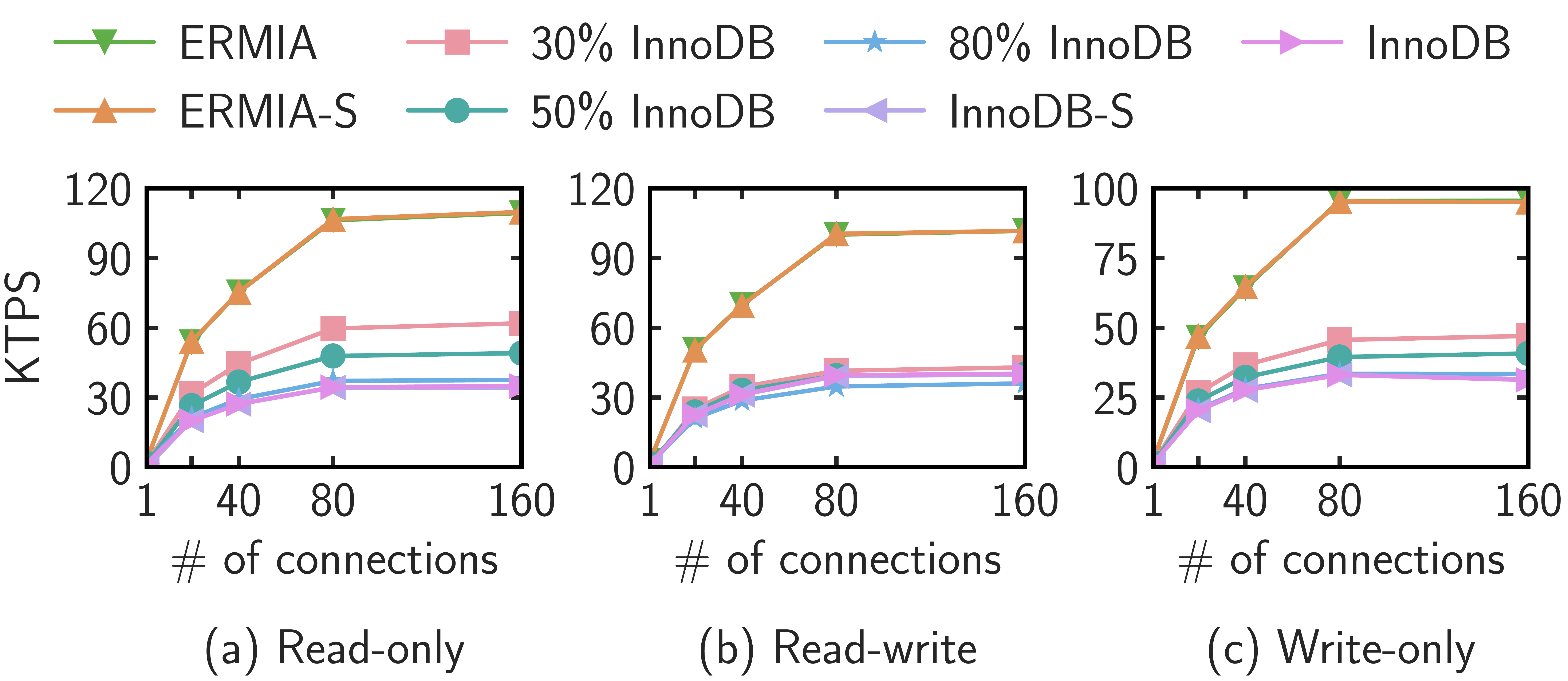}
\caption{\textmd{Throughput under storage-resident microbenchmarks.
Performance improves with more accesses in ERMIA.}}
\label{fig:microbenchmark-storage}
\end{figure}

InnoDB is more heavyweight, so more accesses in it should lower performance, e.g., transactions with 30\% InnoDB accesses should perform better than those that only access InnoDB. 
However, Figures~\ref{fig:microbenchmark-memory}(a)--(b) show the opposite: 
\innodbm outperforms the cross-engine \texttt{30--80\% InnoDB}. 
The reason is two-fold. 
First, ERMIA writes a commit log record for read-only transactions. 
So with more ERMIA accesses, CSR becomes larger and slower to access. 
This is non-negligible for read-intensive workloads, which are very lightweight in ERMIA. 
Second, under \innodbms, CSR is very small and only maintains one mapping as we mentioned earlier. 
However, under \texttt{30--80\% InnoDB}, more ERMIA accesses lead to more mappings in CSR, which then becomes more expensive to query. 
The memory-resident write-only workload follows the expectation in Figure~\ref{fig:microbenchmark-memory}(c), although the difference is not significant due to InnoDB's low raw performance. 
As the workload becomes storage-resident, \skeena's overhead becomes negligible, with more ERMIA accesses leading to higher performance: 
in Figure~\ref{fig:microbenchmark-storage}, \texttt{30\% InnoDB} is up to 75\%/40\% faster than \innodb for read-only/write-only workloads.

\begin{figure}[t]
	\includegraphics[width=\columnwidth]{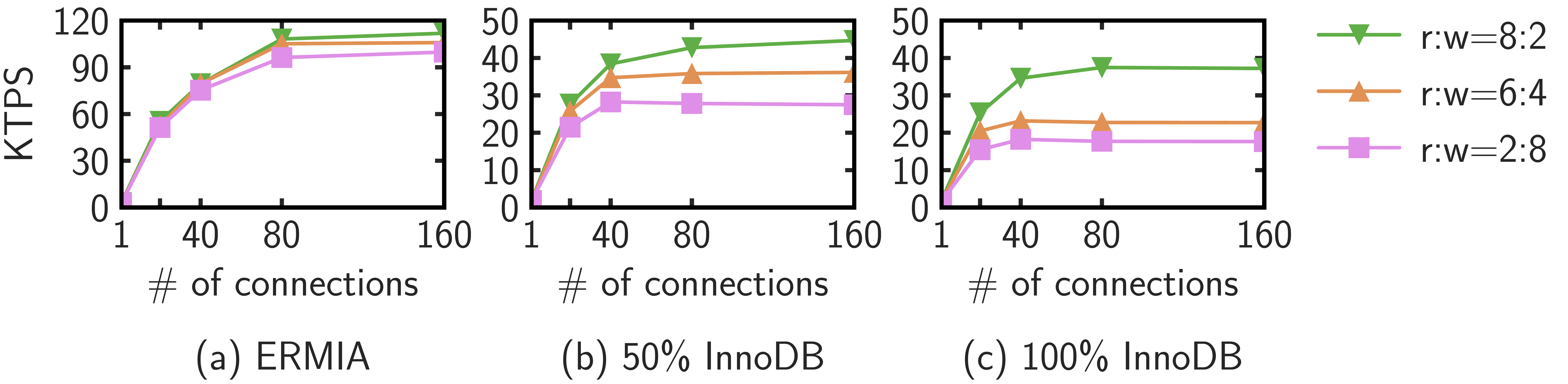}
	\caption{\edit{Throughput under storage-resident microbenchmarks with different read/write ratios.}}
	\label{fig:rw}
\end{figure}

\begin{figure}[t]
	\includegraphics[width=\columnwidth]{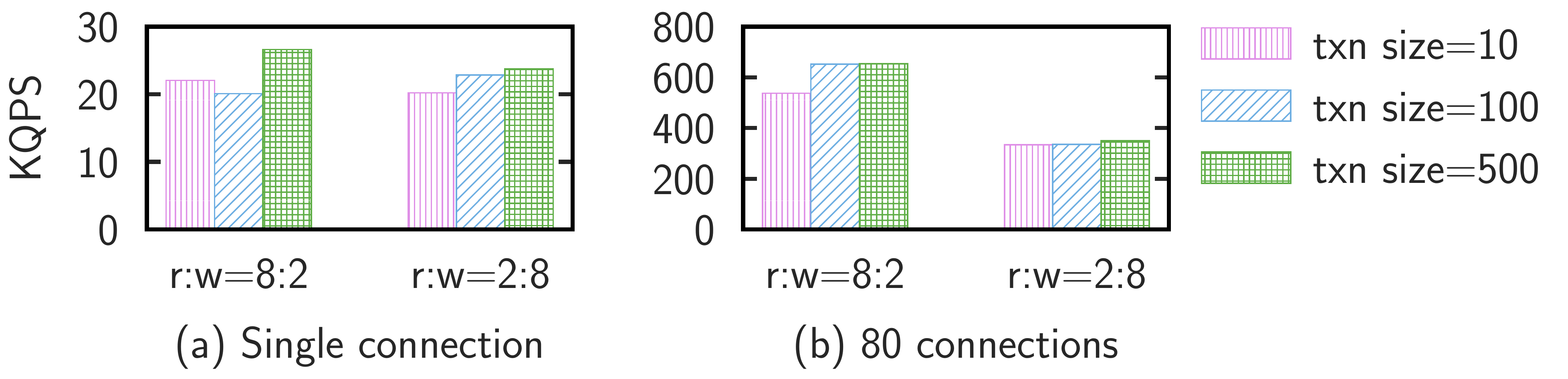}
	\caption{\edit{Throughput under storage-resident microbenchmarks with 50\% InnoDB accesses and different transaction sizes.}}
	\label{fig:tx-size}
\end{figure}

\subsection{Impact of Transaction Size and Mix}
We quantify the impact of transaction sizes along with the impact of read/write mixes. 
We first show the throughput of the storage-resident workload with 50\% InnoDB accesses under varying transaction sizes and read/write ratios. 
In Figure~\ref{fig:tx-size}, longer transactions can lower performance, but do not delay CSR index recycling as all the threads start and commit transactions at roughly the same pace. 
We also ran experiments with a mix of long and short transactions, where a fixed 0--20\% of connections only run long transactions with 500 queries.
We observe the number of indexes increases by $\sim$50 per second (with 1000 entries per index). 
Once the number of indexes reaches a configurable threshold (1000 in our experiments), recycling kicks in and works well across varying percentages of long-running transactions.
Neither long-running transactions nor capacity and threshold settings affect QPS in a noticeable way. 
Figure~\ref{fig:rw} shows the throughput under different read/write ratios. 
With more writes, throughput drops as handling writes is more complex. 
The impact on ERMIA is very small in Figure~\ref{fig:rw}(a). 
When the read ratio drops from 80\% to 60\%, under \texttt{100\% InnoDB}, performance can drop by up to $\sim$30\%. 
Cross-engine transactions in Figure~\ref{fig:rw}(b) with 50\% InnoDB accesses still have the advantage over \texttt{100\% InnoDB}. 

We use short transactions to explore the overhead of \skeena and transaction management. 
Each transaction issues two queries (one per engine for cross-engine cases). 
In Figure~\ref{fig:2ops}, 
when the transaction only accesses ERMIA, the performance remains similar across all workloads, as all data is in memory and the relative speed difference for read/write in ERMIA is small. 
For \texttt{100\% InnoDB}, with writes, the performance drops by up to $\sim$25\%. 
The cross-engine \texttt{50\% InnoDB} exhibits the lowest performance due to extra time spent on CSR and the commit protocol. 
However, it is only slightly slower than \texttt{100\% InnoDB} as it is more heavyweight to process writes in InnoDB than accessing CSR which is purely in-memory. 

\subsection{Skewed Accesses}
We test the storage-resident workload with different ERMIA/InnoDB accesses mixes with 80\% of read and 20\% of write per transaction. 
As Figure \ref{fig:skewness} shows, throughput remains similar across different skewness for ERMIA-only cases. 
The result seems counter-intuitive: a skewed workload leads to a smaller footprint which should change performance.  
The smaller footprint may lead to higher contention on locks/latches (lower performance).  
It could also give better buffer pool and CPU cache locality (higher performance). 
Our profiling results show that the actual CPU time spent on ERMIA is less than 5\%.  
The remaining $>95\%$ of CPU time is taken by MySQL's SQL and networking layers, overshadowing the effect brought by the smaller footprint. 
As we add more accesses to InnoDB, 50\% and 100\% InnoDB show in general higher performance but not by a lot. 
The main reason is the storage stack's overhead starts to dominate once we access InnoDB tables. 
Skewness (smaller footprint) therefore does not show an obvious impact. 

\begin{figure}[t]
	\includegraphics[width=0.9\columnwidth]{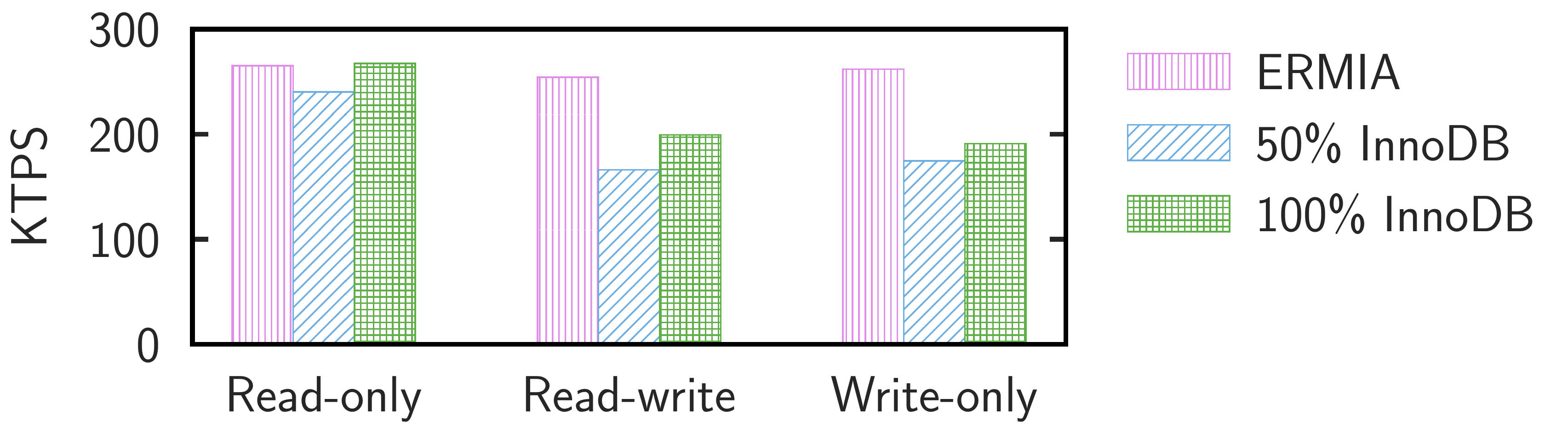}
	\caption{\edit{Throughput under memory-resident benchmarks with short transactions (two queries) under 80 connections.}}
	\label{fig:2ops}
\end{figure}

\begin{figure}[t]
	\includegraphics[width=\columnwidth]{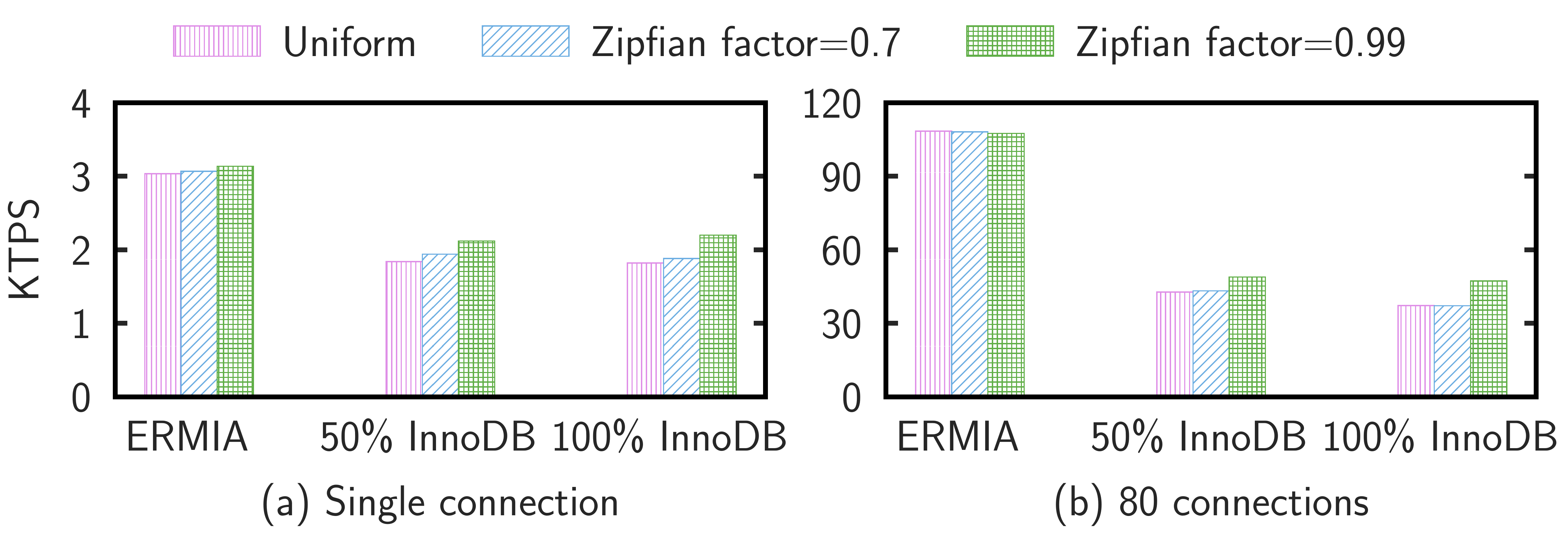}
  \caption{\edit{Throughput of cross-engine read-write transactions under storage-resident microbenchmarks and varying skewness.}}
	\label{fig:skewness}
\end{figure}

\begin{table}[t]
	\centering
  \caption{\edit{Throughput (TPS) of a storage-resident cross-engine workload with 50\% InnoDB accesses under varying buffer pool hit ratios.}} 
	\label{tbl:ssd}
	\begin{tabular}{@{}lllll@{}}
		\toprule
		\bf Number of Connections & \bf 100\% & \bf 99\% & \bf 90\% & \bf 70\%   \\ \midrule
		\bf 1 & 1,973 & 1,972 & 1,939 & 1,866 \\
		\bf 80 & 36,749 & 36,186 & 35,414 & 28,369 \\ \bottomrule
	\end{tabular}
\end{table}

\subsection{Impact of Slower Storage}
Now we run experiments in a more realistic environment that uses an actual SSD to store table and log data. 
We use the storage-resident setting and vary the buffer pool size such that the hit ratio is between 70\% and 100\%. 
In this experiment, to stress \skeena's CSR structure, each 10-record transaction accesses 5/5 records in InnoDB/ERMIA with 80\% reads and 20\% writes. 
As shown in Table~\ref{tbl:ssd}, throughput under a single connection remains stable across different hit ratios because the working set is relatively small which gives better locality. 
With more (80) connections, throughput drops as the hit ratio drops because more I/Os on the SSD are needed. 
Our profiling results (details not shown here for brevity) indicate that across all cases \skeena occupies less than 5\% of total CPU time, most of which was on accessing the engine and the SQL layer. 

\begin{figure}[t]
	\centering
	\includegraphics[width=\columnwidth]{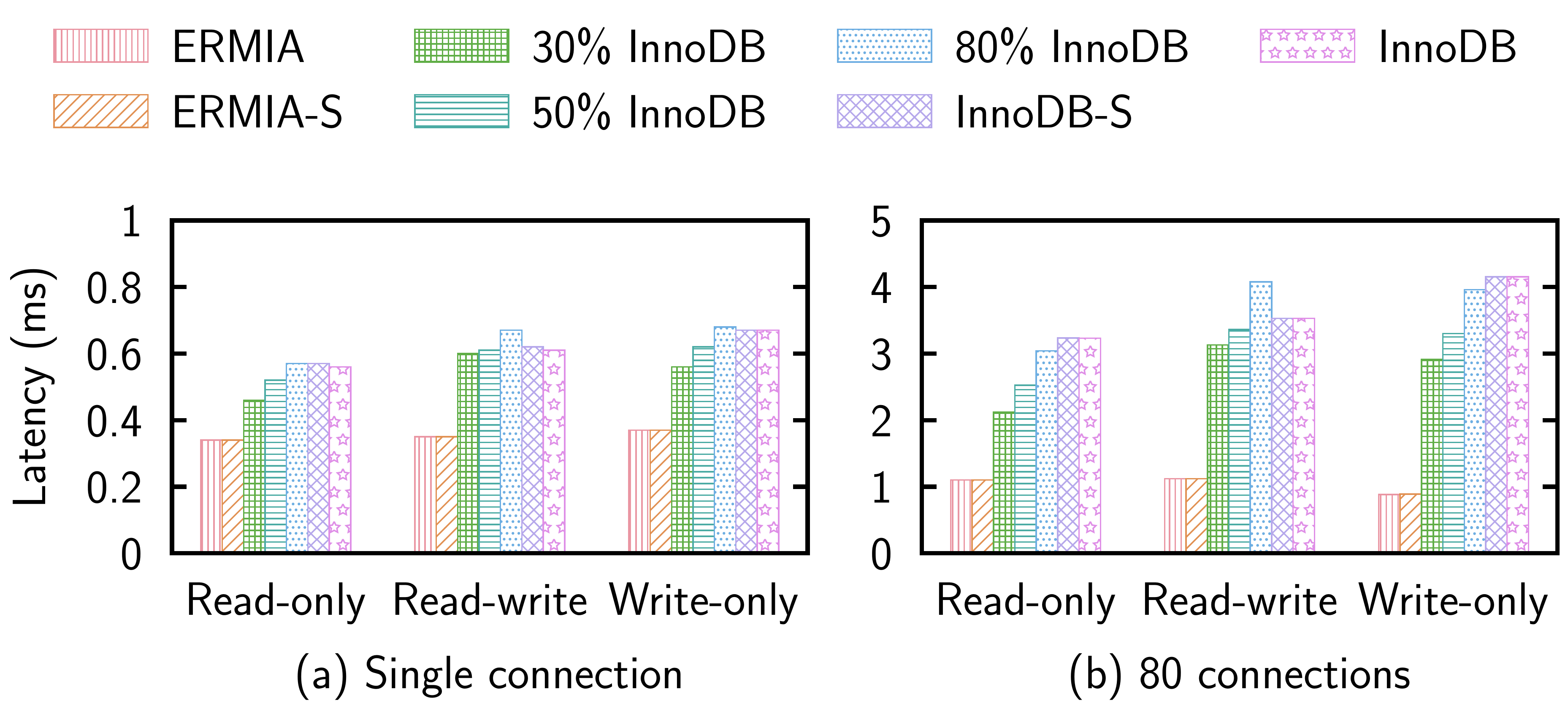}
	\caption{\textmd{95th percentile latency at a single connection and 80 connections under various storage-resident workloads.}}
	\label{fig:microbenchmark-latency}
\end{figure}

\subsection{Transaction Latency}
Our final microbenchmark explores how \skeena impacts transaction latency using the storage-resident microbenchmarks. 
We run the same microbenchmarks done in Section~\ref{subsec:cross-perf} and present results obtained at one and 80 connections where the system is almost idle and fully saturated, respectively. 
In Figure~\ref{fig:microbenchmark-latency}, at 80 connections the absolute latency unavoidably increases. 
But in all cases \skeena does not increase latency noticeably for single-engine transactions: 
\ermia and \ermias do not use CSR, while the overhead for \innodbs is a constant. 
As expected, latency increases proportionally with more accesses in InnoDB. 

\subsection{Effectiveness of Cross-Engine Transactions}
\label{subsec:cross-effect}
Realistic workloads may use cross-engine transactions to improve performance (using a main-memory engine) and/or reduce storage cost with the storage-centric engine. 
This is done by placing different tables in different engines. 
We use TPC-C to explore this aspect; we observed similar trends for memory- and storage-resident setups, so we only show the results from the storage-resident setup. 
As shown in Figure~\ref{fig:tpcc-heatmap-mix}, we start with all tables in InnoDB (bottom) and gradually move them to ERMIA (up). 
Overall, throughput improves with more tables in ERMIA, but performance does not change much until the \texttt{New-Orders} table is moved to ERMIA ($\sim$10$\times$ faster than \texttt{100\% InnoDB}).

Placing \texttt{New-Orders} in ERMIA should be the key to improving overall throughput, but it is unclear how each transaction benefits from this. 
We further run individual TPC-C transactions (instead of the full mix) under a variant that only places \texttt{New-Orders} in ERMIA (leaving the rest in InnoDB). 
Figure~\ref{fig:tpcc-heatmap-gap} compares its throughput (denoted as \texttt{+New-Orders}) to other variants, among which \texttt{++Orders} and \texttt{++New-Orders} refer to the corresponding rows in Figure~\ref{fig:tpcc-heatmap-mix}, respectively. 
In Figures~\ref{fig:tpcc-heatmap-gap}(a)--(b) and \ref{fig:tpcc-heatmap-gap}(d)--(e), placing \texttt{New-Orders} in ERMIA does not affect the New-Order, Payment, Stock-Level and Order-Status transactions. 
This is because Stock-Level, Order-Status and Payment do not access \texttt{New-Orders}, and the New-Order transaction only inserts one row into \texttt{New-Orders}. 
However, the Delivery transaction intensively accesses \texttt{New-Orders} with range scan and aggregation operations. 
In Figure~\ref{fig:tpcc-heatmap-gap}(c), placing \texttt{New-Orders} in ERMIA accelerates Delivery by $\sim$30$\times$ as InnoDB has to hold locks for records to be deleted. 
Therefore, improvement on Delivery is the main reason for the improved overall performance. 

These results show that table placement affects different transactions in different scales. 
We analyze this effect in Figure~\ref{fig:tpcc-heatmap-tx} using individual transactions at 50 connections. 
Tables are gradually moved from InnoDB to ERMIA. 
As Figure~\ref{fig:tpcc-heatmap-tx} shows, placing \texttt{Customer} in ERMIA alone allows the Payment and Order-Status transactions to perform 6--7$\times$ better, whereas the Stock-Level transaction benefits the most when the \texttt{Stock} table is placed in ERMIA. 

\begin{figure}[t]
	\centering	
	\includegraphics[width=\columnwidth]{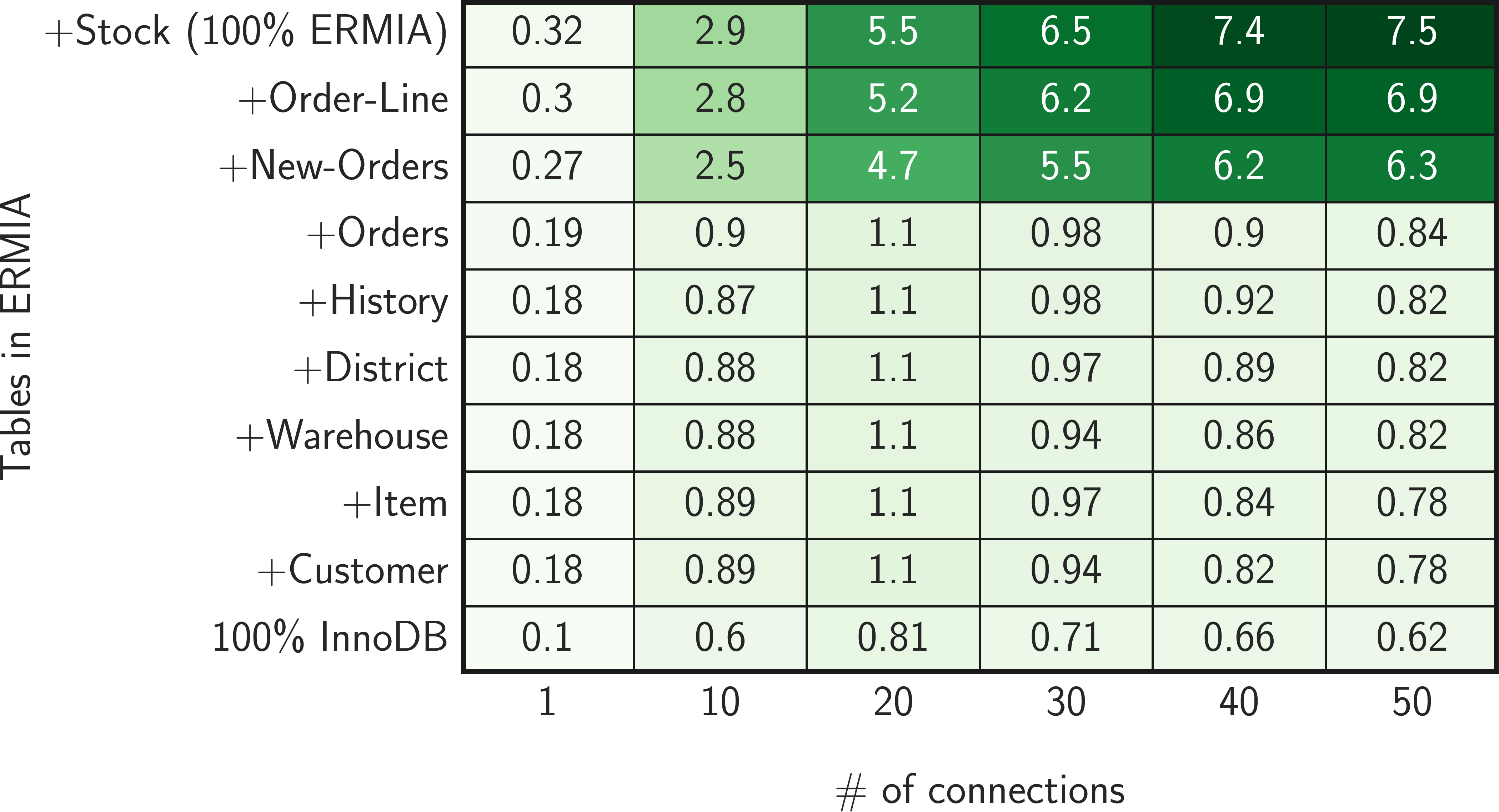}
	\caption{\textmd{TPC-C throughput (kTPS) with tables gradually placed in ERMIA (cumulatively from bottom up). 
  Placing \texttt{New-Orders} in ERMIA is the key to improve overall mix performance.}
  }
	\label{fig:tpcc-heatmap-mix}
\end{figure}

\begin{figure*}[t]
	\includegraphics[width=\textwidth]{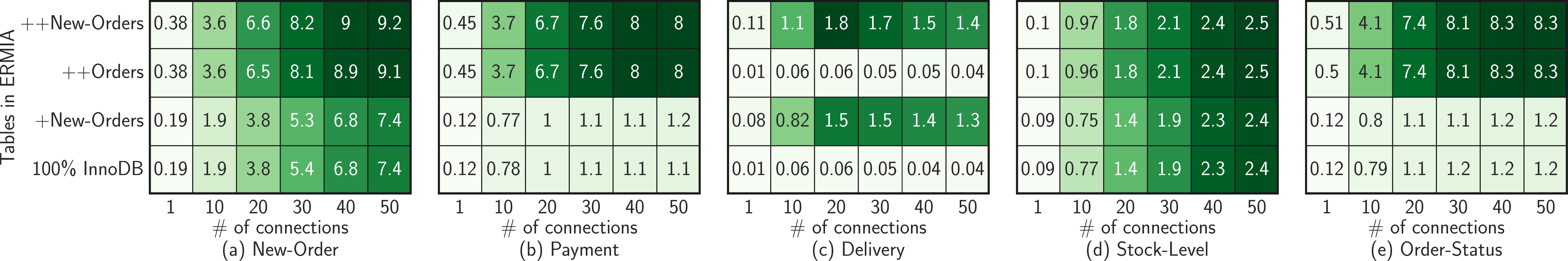}
	\caption{\textmd{Throughput (kTPS) of individual TPC-C transactions when only \texttt{New-Orders} is placed in ERMIA (\texttt{+New-Orders}) compared to \texttt{100\% InnoDB} and two other variants that place cumulatively up to \texttt{Orders} and \texttt{New-Orders} in ERMIA.}} 
	\label{fig:tpcc-heatmap-gap}
\end{figure*}

\begin{figure}[t]
	\includegraphics[width=\columnwidth]{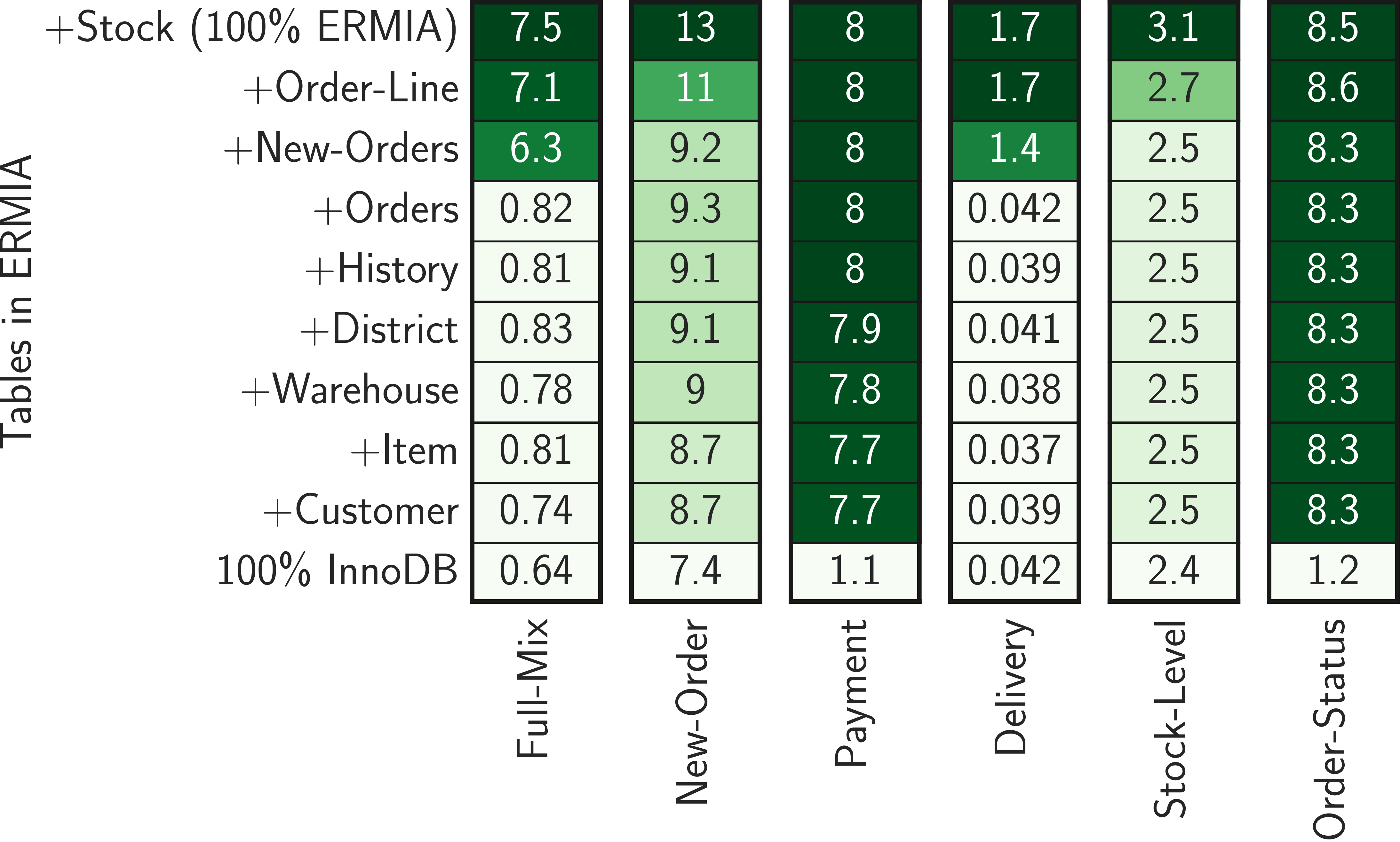}
	\caption{\textmd{Throughput (kTPS) of TPC-C mix and individual transactions by varying table placement at 50 connections.}}
	\label{fig:tpcc-heatmap-tx}
\end{figure}

\textbf{Recommended End-to-End Cross-Engine TPC-C.}
Based on the previous results, we recommend three schemes: 
\begin{itemize}[leftmargin=*]\setlength\itemsep{0em}
\item \texttt{New-Order-Opt}: The \texttt{Customer} and \texttt{Item} tables are placed in ERMIA to optimize the New-Order transaction. 
\item \texttt{Payment-Opt}: Only \texttt{Customer} is placed in ERMIA to optimize the Payment transaction, which intensively accesses \texttt{Customer}. 
\item \texttt{Archive}: All the tables except \texttt{History} are placed in InnoDB, leveraging its cheaper storage cost compared to ERMIA. 
\end{itemize}

The first two schemes aim to optimize database accesses with select use of main-memory tables, while \texttt{Archive} attempts to reduce storage cost of in-memory databases using a traditional engine.
Figure~\ref{fig:tpcc-disk-resident} shows how they compare to baselines that place all tables in ERMIA and InnoDB. 
\texttt{New-Order-Opt} and \texttt{Payment-Opt} improve the performance of the affected transactions compared to \innodb. 
Since \texttt{Archive} executes almost fully in ERMIA, its performance overlaps with \ermia because \texttt{History} is never queried and only occupies less than 600MB of space. 
In reality, such workloads can run for much longer and accumulate much more data; placing it in InnoDB can drastically reduce storage cost as main memory is much more expensive than SSDs and disks. 

\textbf{Impact on Abort Rate.}
As we discussed in Section~\ref{sec:skeena}, \skeena can cause transactions to abort if they cannot find appropriate snapshots or fail commit check. 
We compare the abort rate using the preferred end-to-end TPC-C schemes. 
The workload runs under the memory-resident setup which exhibits low contention (each transaction works on a fixed home warehouse) so that snapshot selection and commit check in \skeena are the main source of aborts. 
As baselines, single-engine \innodb and \ermia exhibit an abort rate of 0.43\% and 0.47\%, respectively. 
With \skeena, \texttt{New-Order-Opt}, \texttt{Payment-Opt} and \texttt{Archive} respectively exhibit an abort rate of 0.61\%, 0.54\% and 0.45\%. 
Across individual runs, we observed up to 0.3\% higher abort rate than single-engine cases. 
The impact on abort rate is more pronounced for the read-write microbenchmarks, where up to $\sim$5\% additional cross-engine transactions are aborted due to \skeena. 
These results corroborate with prior work~\cite{IncrementalDSI} that the impact on abort rate is very small for realistic workloads. 

\begin{figure}[t]
	\centering
	\includegraphics[width=0.9\columnwidth]{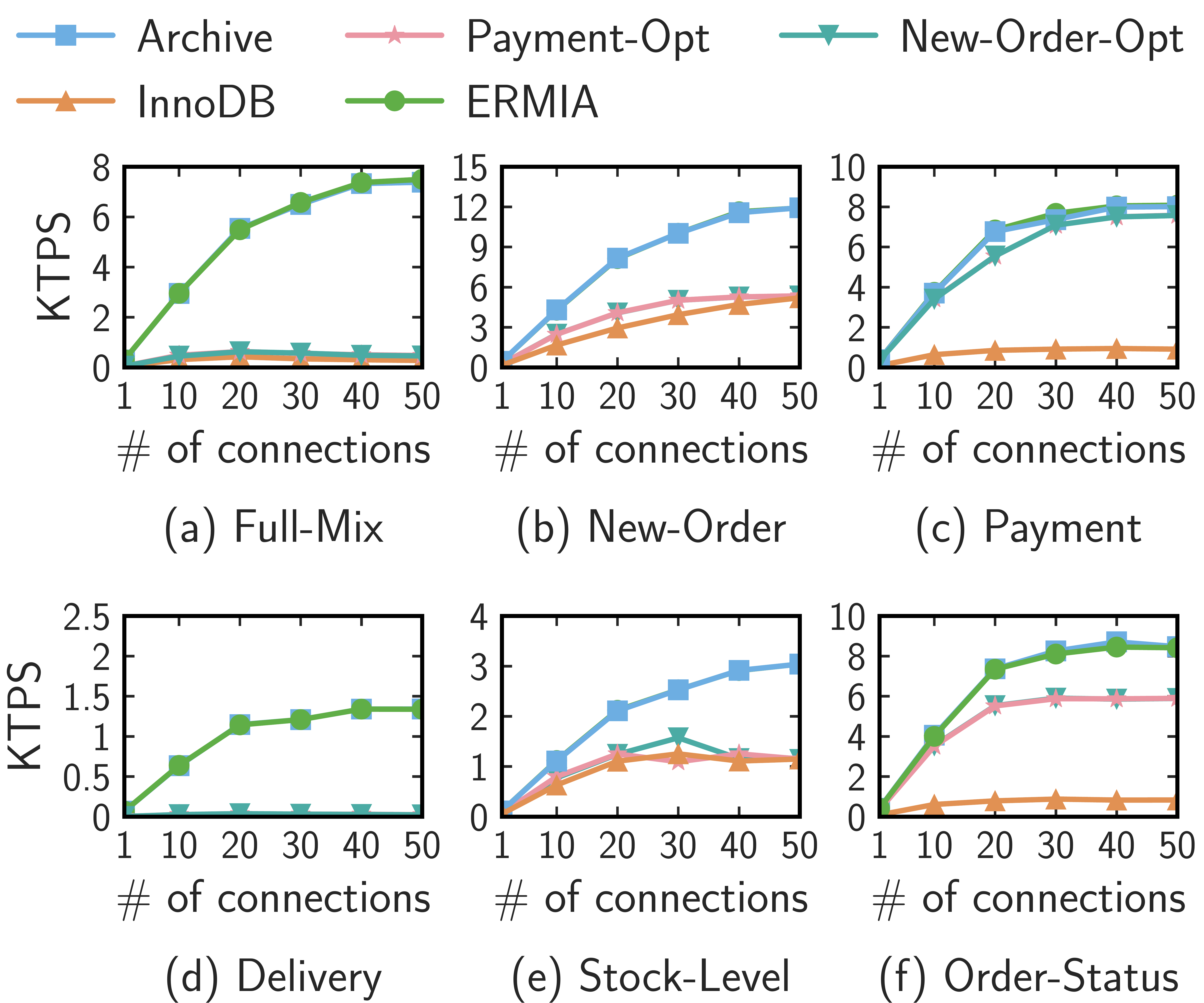}
	\caption{\textmd{TPC-C throughput under select table placement schemes that optimize for different application scenarios.}}
	\label{fig:tpcc-disk-resident}
\end{figure}

%% file: 6-related-work.tex
\section{Related Work}
\label{sec:related}
Our work builds upon rich literature on federated systems, distributed and replicated SI systems, and modern database engines.

\textbf{Federated and Polystore Systems.}
Guaranteeing serializability is challenging due to the full autonomy and heterogeneity of member systems. 
Georgakopoulos et al.~\cite{MDBSForcedUpdates} proposed a ticket method and Superdatabase~\cite{SuperDB} exports transaction ordering for consistent updates. 
Breitbart et al.~\cite{MDBSUpdate} proposed a protocol that detects cycles for consistency and resolving deadlocks.
Myriad~\cite{MYRIAD,MYRIADDesign} uses 2PL and 2PC for serializability and atomicity. 
Recent polystores~\cite{bigdawg,agrawal2018rheem,alotaibi2019towards} mainly focus on analytics, while our focus is OLTP. 

\textbf{DSI and Replication.}
Prior work identified the anomalies that would lead to inconsistent snapshots~\cite{FederatedSIIntegration}.  
Binnig et al.~\cite{IncrementalDSI} further identified more anomalies and proposed correctness conditions. 
\skeena is based on these results. 
Generalized SI~\cite{GeneralizedSI} allows transactions to use older snapshots in replicated databases. 
\skeena shares the similar property by selecting snapshots using CSR. 

\textbf{Modern Database Engines.}
Multicore CPUs and large DRAM have led to numerous memory-optimized engines~\cite{MMDBOverview,Cicada,Silo,FOEDUS,ERMIA,Deuteronomy,Hekaton,HStore,Hyper}. 
They feature new designs on indexing, concurrency control and logging protocols that drastically improve performance. 
Our work explores one of the possible ways to adopt them in practice. 
Some recent systems~\cite{LeanStore,Umbra} efficiently leverage modern fast SSDs to approach in-memory performance while keeping storage cost lower than main-memory systems.
It is interesting future work to compare them with cross-engine systems.

%% file: 7-conclusion.tex
\section{Conclusion}
Cross-engine transactions can be very useful in modern fast-slow multi-engine systems, but are poorly supported with various limitations. 
This paper proposes \skeena, a holistic approach to efficient and consistent cross-engine transactions. 
\skeena consists of a cross-engine snapshot registry (CSR) that tracks snapshots and a commit protocol for multi-engine systems.
\skeena can be easily adopted by real systems, as shown by our experience with MySQL. 
Evaluation on a 40-core server shows that \skeena incurs negligible overhead and maintains the benefits of memory-optimized engines. 
\label{sec:conclusion}